\documentclass[10pt]{amsart} % article/book/letter/amsart

\usepackage[utf8]{inputenc}
\usepackage[T1]{fontenc}	% accents
\usepackage[french,english]{babel}

\usepackage{lmodern}

\usepackage{graphicx}	% insertion des images / graphiques
\usepackage[a4paper,plainpages=false,colorlinks,linkcolor=bleuFonce,citecolor=rougeFonce,urlcolor=vertFonce,breaklinks]{hyperref}

\usepackage{placeins}
\usepackage{float} % pour mettre H dans les figures

\usepackage{color}
\definecolor{vertFonce}{rgb}{0,0.5,0}
\definecolor{numLignes}{rgb}{0.17,0.57,0.7}	%{43,145,175}
\definecolor{gris}{rgb}{0.5,0.5,0.5}
\definecolor{grisFonce}{rgb}{0.2,0.2,0.2}
\definecolor{orange}{rgb}{1,0.65,0.31}		%{255,167,79}
\definecolor{orangeFonce}{rgb}{1,0.4,0}
\definecolor{bleuFonce}{rgb}{0,0,0.4}
\definecolor{rougeFonce}{rgb}{0.3,0,0}
\definecolor{rougeWord}{rgb}{0.5,0,0}
\definecolor{vertClair}{rgb}{0.8,1,0.8}
\definecolor{rougeClair}{rgb}{1,0.5,0.5}

\usepackage{pict2e}
\usepackage{multido}

\usepackage{amsfonts,amssymb,amsthm,amsmath} 
\usepackage{dsfont}
\usepackage{mathrsfs}
\usepackage{yfonts}
 
\newtheorem{lem}{Lemma}[section]
\newtheorem{thm}{Theorem}[section]

\newtheorem{prop}{Proposition}[section]

\newtheorem{remark}{Remark}[section]

\newcommand{\step}[1]	{\paragraph{\itshape\bfseries Step #1.}}

\usepackage{stmaryrd}
\newcommand		{\subsetArrow}	{\mathrel{\ooalign{$\subset$\cr%
\hidewidth\raise-.087ex\hbox{$_\shortrightarrow\mkern-1.5mu$}\cr}}}
\newcommand		{\subsetarrow}	{\mathrel{\ooalign{$\subset$\cr%
\hidewidth\raise-1.45ex\hbox{$\vec{}\mkern6mu$}\cr}}}

\newcommand		{\ssi}		{\Leftrightarrow}

\newcommand		{\N}		{\mathbb N}			% naturels
\newcommand		{\R}		{\mathbb R}			% real numbers
\newcommand     {\cM}		{\mathcal M}		% Mesures bornées
\renewcommand	{\L}		{\mathcal L}		% Applications linéaires ou loi d'une variable aléatoire
\newcommand		{\B}		{\mathscr B}		% Applications linéaires continues
\newcommand		{\cU}		{\mathcal U}
\newcommand		\sfT		{\mathsf T}			% Opérateur de transport

\newcommand		\sfX		{\mathsf X}			% Opérateur d'échange
\newcommand		{\lt}			{\left}				%
\newcommand		{\rt}			{\right}			%
\renewcommand	{\(}			{\lt(}
\renewcommand	{\)}			{\rt)}
\newcommand		{\lal}			{\langle}			%
\newcommand		{\ral}			{\rangle}			%

\newcommand		{\weight}[1]	{\lt\lal #1\rt\ral}	% <x>

\newcommand		{\com}[1]		{\lt[{#1}\rt]}		% commutator
\newcommand		{\il}			{[\![}				% integer left bracket
\newcommand		{\ir}			{]\!]}				% integer right bracket
\newcommand		{\Int}[1]		{\il #1 \ir}
\newcommand		{\n}[1]			{\lt|{#1}\rt|}
\newcommand		{\nrm}[1]		{\lt\|{#1}\rt\|}
\newcommand		{\Nrm}[2]		{\lt\|{#1}\rt\|_{#2}}

\renewcommand		{\d}		{\mathrm{d}}		% différentielle
\newcommand			{\dd}		{\,\d}				% différentielle (intg)
\newcommand			{\dpt}		{\partial_t}
\newcommand			{\dps}		{\partial_s}

\newcommand			{\dt}		{\frac{\d}{\d t}}	% dérivée temporelle
\newcommand			{\ddt}[1]	{\frac{\d #1}{\d t}}

\newcommand			{\Dx}		{\nabla_x}
\newcommand			{\Dv}		{\nabla_\xi}
\newcommand			{\Id}		{\mathrm{1}}		% Identity operator
\DeclareMathOperator{\cF}		{\mathcal{F}}		% Fourier transform
\DeclareMathOperator{\sign}		{sgn}				% signe
\DeclareMathOperator{\tr}		{Tr}				% Trace
\DeclareMathOperator{\diag}		{diag}

\newcommand		{\F}[1]			{\cF\!\( #1 \)}		% Fourier transform
\newcommand		{\Sign}[1]		{\sign\!\( #1 \)}	% signe
\newcommand		{\Tr}[1]		{\tr\!\( #1 \)} % Trace
\newcommand		{\Diag}[1]		{\diag\!\( #1 \)}

\newcommand		{\intd}			{\int_{\R^d}}

\newcommand		{\iintd}		{\iint_{\R^{2d}}}

\newcommand		{\sumj}			{\sum_{j\in J}}

\newcommand		{\ii}			{\mathrm{i}}	% i\in\N (et pas i^2 = -1)
\newcommand		{\init}			{\mathrm{in}}
\newcommand		{\loc}			{\mathrm{loc}}

\newcommand		{\fb}			{\mathfrak b}
\newcommand		{\eps}			{\varepsilon}
\newcommand		{\Eps}			{\mathcal{E}}
\newcommand		{\cC}			{\mathcal{C}}

\usepackage{upgreek}
\renewcommand	{\r}		{\op}			% noyau de \op
\newcommand		{\op}		{\boldsymbol{\rho}}	% opérateur densité
\newcommand		{\opm}		{\boldsymbol{m}}	% opérateur densité
\newcommand		{\bra}[1]	{\langle #1 |}
\newcommand		{\ket}[1]	{| #1 \rangle}
\newcommand		{\w}		{w}					% Wigner transform
\newcommand		{\wh}		{\w_\hbar}			% Wigner transform

\newcommand		{\weyl}		{\op^{W\!}}				% Weyl transform
\newcommand		{\Weyl}[1]	{\weyl_1\!\( #1 \)}
\newcommand		{\Weylh}[1]	{\weyl_\hbar\!\( #1 \)}

\newcommand		{\opp}		{\boldsymbol{p}}

\newcommand		{\Dh}		{\boldsymbol{\nabla}}			% quantum d_x
\newcommand		{\Dhx}[1]	{\Dh_{\!x} #1}		% quantum d_x
\newcommand		{\Dhv}[1]	{\Dh_{\!\xi} #1}	% quantum d_xi
\newcommand		{\Dhvv}[1]	{\Dh^2_{\!\xi} #1}	% quantum d_xi
\newcommand		{\floor}[1]		{\lt\lfloor{#1}\rt\rfloor}
\title[Strong semiclassical limit]{Strong semiclassical limit from Hartree and Hartree-Fock to Vlasov-Poisson equation}
\author{Laurent Lafleche$^1$, Chiara Saffirio$^2$}
\thanks{$^1$ Department of Mathematics, The University of Texas at Austin, Austin, TX 78712, USA, {\tt lafleche@math.utexas.edu}}
\thanks{$^2$ Department of Mathematics and Computer Science, University of Basel, 4051 Basel, Switzerland, {\tt chiara.saffirio@unibas.ch}}
\begin{document}

\begin{abstract} We consider the semiclassical limit from the Hartree to the Vlasov equation with general singular interaction potential including the Coulomb and gravitational interactions, and we prove explicit bounds in the strong topologies of Schatten norms. Moreover, in the case of fermions, we provide estimates on the size of the exchange term in the Hartree-Fock equation and also obtain a rate of convergence for the semiclassical limit from the Hartree-Fock to the Vlasov equation in Schatten norms. Our results hold for general initial data in some Sobolev space and any fixed time interval.
\end{abstract}

\maketitle

\bigskip


%----------  Table of Contents  ----------
\renewcommand{\contentsname}{\centerline{Table of Contents}}
\setcounter{tocdepth}{2}	% profondeur du commentaire
\tableofcontents
%\addcontentsline{toc}{chapter}{Table of Contents}

%% ********************  Contenu  ********************

%----------  Introduction  ----------
\bigskip
\section{Introduction}\label{sec:intro}

	The Vlasov equation is a kinetic equation describing the time evolution of the probability density of particles in interaction, such as particles in a plasma or in a galaxy. The problem of deriving this equation from the dynamics of $N$ quantum interacting particles in a joint mean-field and semiclassical approximation is a classical question in mathematical physics and the first rigorous results were obtained in the '80s (Cf. \cite{narnhofer_vlasov_1981, spohn_vlasov_1981}).
	
	We study here the semiclassical limit from the Hartree and Hartree-Fock equations towards the Vlasov equation, i.e. the limit corresponding to a change of scaling implying that the Planck constant $h$ becomes negligible. For any fixed time interval, we obtain quantitative Schatten norm estimates between the solutions of the quantum equations (Hartree and Hartree-Fock) and the Weyl quantization of the solution of the Vlasov equation. In particular, it implies the convergence of the Wigner transform of the quantum equations towards the solution of the Vlasov equation.
	
\subsection{Context and state of the art}
	
	\subsubsection{Vlasov equation} The Vlasov equation is a nonlinear transport equation for the probability density $f:\R_+\times\R^d\times\R^d\to\R$
	\begin{equation}\label{eq:Vlasov}
		\partial_t f + \xi\cdot\nabla_x f + E\cdot\Dv f=0\,,
	\end{equation}
	where $t\in\R_+$ denotes the time variable, $x\in\R^d$ denotes the space variable and $\xi\in\R^d$ denotes the momentum variable. In the above equation, $E := -\nabla K* \rho_f$ is the self induced mean-field force field created by the pair interaction potential $K:\R^d\to\R$ through the formula
	\begin{equation*}
		-(\nabla K*\rho_f)(t,x) = -\intd \nabla K(x-y)\,\rho_f(t,y)\dd y,
	\end{equation*}
	where $\rho_f$ is the spatial density associated to $f$, namely 
	\begin{equation*}
		\rho_f(t,x) = \intd f(t,x,\xi)\dd \xi\,.
	\end{equation*}
	When $K$ is the Green's function of the Laplace operator, Equation~\eqref{eq:Vlasov} is called the Vlasov-Poisson system, because $K$ can be obtained as a solution to the Poisson equation $-\Delta K=\rho_f$, thus linking the Vlasov equation to the Poisson equation. In this case, in dimension $3$, $K$ corresponds to the Coulomb potential
	\begin{equation*}
		K(x) = \frac{1}{4\pi\n{x}},
	\end{equation*}
	but our method allows to consider more general attractive and repulsive potentials. To simplify the presentation, we will look at homogeneous potentials of the form $K(x) = \pm\n{x}^{-a}$ or at $K(x) = \pm\ln(\n{x})$ and we will then indicate how to generalize our results to a class of Sobolev spaces (see Subsection~\ref{subsec:generalization_K}).

	The well-posedness of the Vlasov equation~\eqref{eq:Vlasov} is due to Dobrushin~\cite{dobrushin_vlasov_1979} for smooth interaction potentials $K\in C^2_c(\R^d)$. Concerning singular interactions, the cases of Coulomb and gravitational potentials have been tackled first in~\cite{iordanskij_cauchy_1961} and~\cite{ukai_classical_1978}, respectively for $d=1$ and $d=2$. In $d=3$, the well-posedness for small data has been proven in~\cite{bardos_global_1985} and later extended to general initial data by Pfaffelmoser~\cite{pfaffelmoser_global_1992} and by Lions and Perthame~\cite{lions_propagation_1991}. In recent years improvements on the conditions of propagation of momenta and on the uniqueness condition have been addressed in~\cite{pallard_moment_2012, pallard_space_2014, desvillettes_polynomial_2015, loeper_uniqueness_2006, miot_uniqueness_2016, holding_uniqueness_2018}. The setting of this paper will be close to the setting of the paper by Lions and Perthame \cite{lions_propagation_1991}, that is the one that better adapt to the comparison with the quantum dynamics because of its Eulerian viewpoint.  

	The Vlasov equation~\eqref{eq:Vlasov} is supposed to emerge as a joint mean-field and semiclassical limit from the dynamics of $N$ interacting quantum particles. This has been first proven in \cite{narnhofer_vlasov_1981, spohn_vlasov_1981}, respectively for analytic and $C^2$ interaction potentials, using the BBGKY approach in the fermionic setting. The case of bosons interacting through a smooth pair potential has been studied in~\cite{graffi_mean-field_2003}, in the mean-field limit combined with a semiclassical limit, through the analysis of the dynamics of factored WKB states.  

	\subsubsection{Hartree and Hartree-Fock equations} It is well known that the many-body dynamics can be approximated in the mean-field limit by the Hartree equation
	\begin{equation}\label{eq:Hartree}
		i\hbar\,\dpt\op = [H,\op],
	\end{equation}
	an evolution equation for the density operator $\op = \op(t)$, a nonnegative bounded operator on the space $L^2(\R^d)$, with $\Tr{\op}=1$. In Equation~\eqref{eq:Hartree}, $\hbar = \frac{h}{2\pi}$ is the reduced Planck constant, and $H$ is the Hamiltonian
	\begin{equation}\label{eq:Hartree-Ham}
		H=-\tfrac{\hbar^2}{2}\,\Delta+K*\rho\,,
	\end{equation}
	where $\Delta$ is the Laplace operator, $K$ is the pair interaction potential, $\rho(x) = \r(x,x)$ the diagonal of the integral kernel of the trace class operator $\op$ and $K*\rho$ is identified with the operator of multiplication by the function $x\mapsto K*\rho(x)$.
	
	In the case of fermions, a more precise mean-field approximation for the many-body quantum dynamics is given by the Hartree-Fock equation
	\begin{equation}\label{eq:HF}
		i\hbar\,\dpt\op = [H_{\text{HF}},\op],
	\end{equation}
	with $H_{\text{HF}} = -\hbar^2\,\Delta+K*\rho-\sfX$, where $\sfX$ is the so called exchange term defined as the operator with integral kernel
	\begin{equation}
		\sfX(x,y) = K(x-y)\,\r(x,y).
	\end{equation}
	
	We recall that the interest in the mean-field regime is due to the fact that many systems of interest in quantum mechanics are usually made of a large number of particles, which typically ranges between $10^2$ and $10^{23}$, while the above equations only describe the behavior of one typical particle in the limit of infinitely many particles. The mathematical literature on this subject is rather extensive. See for example \cite{bardos_weak_2000, erdos_derivation_2001, bardos_derivation_2002, frohlich_mean-field_2009, rodnianski_quantum_2009, grillakis_second-order_2010, pickl_simple_2011, chen_rate_2011, kuz_rate_2015, golse_mean_2016, mitrouskas_bogoliubov_2019, golse_schrodinger_2017, golse_derivation_2018, golse_empirical_2019, chen_rate_2018} for the case of bosons and \cite{elgart_nonlinear_2004, frohlich_microscopic_2011, benedikter_mean-field_2014, benedikter_mean-field_2016, bach_kinetic_2016, petrat_new_2016, porta_mean_2017, petrat_hartree_2017, saffirio_mean-field_2018} for the case of fermions.

	\subsubsection{Semiclassical limit} The Hartree and Hartree-Fock equations are quantum models. It is therefore natural to investigate their semiclassical limit as $\hbar\to 0$. First results in this direction provide the convergence from the Hartree dynamics towards the Vlasov equation in abstract sense, without rate of convergence and in weak topologies, but including the case of singular interaction potentials, such as the Coulomb interaction (Cf. \cite{lions_sur_1993, markowich_classical_1993, gasser_semiclassical_1998, figalli_semiclassical_2012}). Explicit bounds on the convergence rate in stronger topologies have been established in~\cite{pezzotti_mean-field_2009, athanassoulis_strong_2011, amour_classical_2013, amour_semiclassical_2013, benedikter_hartree_2016, golse_schrodinger_2017}. They all deal with smooth interaction potentials. More recently, the case of singular interactions, including the Coulomb potential, has been considered in~\cite{lafleche_propagation_2019, lafleche_global_2019}, where the convergence from the Hartree to the Vlasov equation is achieved in weak topology using quantum Wasserstein-Monge-Kantorovich distance, providing explicit bounds on the convergence rate. In strong topology (trace norm and Hilbert-Schmidt norm) explicit bounds on the convergence from the Hartree dynamics to the Vlasov equation with inverse power law of the form $K(x)=|x|^{-a}$ with $a\in(0,1/2)$ have been proven in~\cite{saffirio_semiclassical_2019}, and a proof that includes the Coulomb potential has been provided in \cite{saffirio_hartree_2020} but under restrictive assumptions on the initial data.
	
	\subsubsection{Key novelties} The aim of this paper is to establish a strong convergence result from both the Hartree and the Hartree-Fock equations towards the Vlasov dynamics for a large class of regular initial states. Our results apply to a wide class of initial data smooth in the semiclassical limit, thus giving a thorough answer to the question of strong convergence of the Hartree equation to the Vlasov system for singular interactions, at least in the case of mixed states converging to smooth solutions of the Vlasov equation.
	 
	With respect to the results present in literature, there are several novelties: apart from the large class of initial data for whose evolution we can establish strong convergence with explicit rate towards the Vlasov equation, our techniques allow to consider inverse power law potentials that are more singular than Coulomb and our methods easily extend to very general non radially symmetric potentials. Moreover, the topology we consider is not only the one induced by the trace or Hilbert-Schmidt norm (as it is for instance in \cite{saffirio_semiclassical_2019}), but the ones induced by semiclassical Schatten norms $\L^p$, for all $p\in[1,\infty)$. These are obtained by a refinement on the estimate for the $\L^p$ norms of the commutator $[K(\cdot-z),\op]$ and a careful analysis of the propagation in time of initial conditions leading to bound the quantity 
	\begin{equation*}
		\Nrm{{\rm diag}\n{\com{\frac{x}{i\hbar},\op}}}{L^p(\R^d)}
	\end{equation*}
	uniformly in $\hbar$, for $p>3$. This requires using kinetic interpolation inequalities as in~\cite{lafleche_propagation_2019} and an extension of the Calder\'{o}n-Vaillancourt theorem for Weyl quantization.
	
	Finally, we extend our results to the Hartree-Fock equation~\eqref{eq:HF}, thus proving the strong convergence of the Hartree-Fock dynamics to the Vlasov equation. As a corollary, we get explicit estimates on the difference between the Hartree and Hartree-Fock dynamics in Schatten norms, thus giving a rigorous proof of the fact that the exchange term in the Hartree-Fock dynamics is subleading with respect to the direct term also when the interaction potential is singular (this was proved in~\cite{benedikter_mean-field_2014} in the case of smooth potentials).
	
	\subsubsection{Open problems} Our work gives good answers to the problem of the semiclassical limit from the Hartree and Hartree-Fock equations to the Vlasov equation with general singular potentials in the context of mixed states. However, a certain number of questions related to the derivation of the Vlasov equation from quantum dynamics remain open:
	\begin{itemize}
	\item[i)] The mean-field limit from a system of $N$ quantum particles interacting through a singular potential in the case of mixed states. Up to our knowledge, this problem is open in both the bosonic and the fermionic setting.
	\item[ii)] In the bosonic setting, where $N$ and $\hbar$ are independent parameters, the joint mean-field and semiclassical limit is an open problem when the interaction is singular. Namely, no uniform convergence in the semiclassical parameter $\hbar$ has been proven so far.
	\item[iii)] We believe our results give optimal bounds on the convergence rate in trace norm $\L^1$. %TODO Can we prove it ?
	The question whether the bounds we obtain for the semiclassical Hilbert-Schmidt norm $\L^2$, and thus the $L^2$ convergence for the associated Wigner functions, are optimal is open. The exact same question can be asked for the bounds in Theorem~\ref{thm:CV_Hartree-Fock} about the convergence of the Hartree-Fock equation to the Vlasov equation. In both cases, we believe the bounds we get are not optimal and there is room for improvements. 
	\end{itemize} 
	
	\subsubsection{Structure of the paper} The paper is structured as follows.
	\begin{itemize}
	\item In Section~\ref{subsec:main_results} we state our main results, and we include additional comments and generalizations in Section~\ref{subsec:discussion}.
	\item In Section~\ref{sec:strategy} we explain our strategy. We introduce a semiclassical notion of regularity (Section~\ref{subsec:gradients}) and then explain our method to get the semiclassical limit by making a comparison with the classical Vlasov dynamics, finding a new stability estimate for the Vlasov system  (Section~\ref{subsec:classical_case}).
	\item Section~\ref{sec:regularity} contains the main results concerning the regularity of the Weyl transform of a solution to the Vlasov equation, that will be crucial to prove the theorems stated in Section~\ref{subsec:main_results}.
	\item Section~\ref{sec:proof_thm_1_2} is devoted to prove Theorem~\ref{thm:CV_Hartree} and Theorem~\ref{thm:Lp_conv}, dealing with the semiclassical limit from the Hartree equation, under the assumption that the regularity proven in Section~\ref{sec:regularity} holds.
	\item In Section~\ref{sec:proof_thm_3} we present the proof of Theorem~\ref{thm:CV_Hartree-Fock} about the semiclassical limit from the Hartree-Fock equation, based on additional estimates on the exchange term.
	\item Two appendices on the propagation of regularity for the Vlasov equation and on basic operator identities complement the paper.  
	\end{itemize}

\subsection{Main results}\label{subsec:main_results}
	 
	\subsubsection{Operators and function spaces} We denote by $L^p = L^p(\R^d)$ the classical Lebesgue spaces, by $L^{p,q} = L^{p,q}(\R^d)$ the classical Lorentz spaces for $(p,q)\in[1,\infty]^2$ (see for example \cite{bergh_interpolation_1976}). In particular $L^{p,p}=L^p$. We define the space of positive and trace class operators by 
	\begin{equation*}
		\L^1_+ := \{\op\in\L(L^2), \op = \op^* \geq 0, \Tr{\op} < \infty\},
	\end{equation*}
	where $\L(L^2)$ denotes the space of linear operators on $L^2$, and the quantum Lebesgue norms (or semiclassical Schatten norms) $\L^p$ by
	\begin{align*}
		\|\op\|_{\L^p} &:= h^{-d/p'} \|\op\|_p = h^{-d/p'} \(\Tr{|\op|^p}\)^\frac{1}{p}.
	\end{align*}
	where $\|\op\|_p$ denotes the usual Schatten norm (i.e. without dependency in $h$) and $p' = \frac{p}{p-1}$ denotes the conjugate exponent.
	
	In this work, we consider the semiclassical limit to solutions of the Vlasov equation with regular data in the sense that the initial condition will be bounded in some weighted Sobolev space. Therefore, we will use the following notation for smooth polynomial weight functions
	\begin{equation*}
		\weight{y} := \sqrt{1+\n{y}^2},
	\end{equation*}
	and for $\sigma\in\N$, we define the spaces $W^{\sigma,p}_k(\R^{2d})$ as the spaces equipped with the norm
	\begin{equation*}
		\Nrm{f}{W^{\sigma,p}_k(\R^{2d})} := \Nrm{\weight{z}^k f(z)}{L^p(\R^{2d})} + \Nrm{\weight{z}^k \nabla_z^\sigma f(z)}{L^p(\R^{2d})},
	\end{equation*}
	where $z = (x,\xi)$ so that $\weight{z}^2 = 1+\n{x}^2+\n{\xi}^2$. We also use the standard notations in the cases $\sigma=0$ or $p=2$
	\begin{align*}
		L^p_k(\R^{2d}) &:= W^{0,p}_k(\R^{2d}), &	H^\sigma_k(\R^{2d}) &:= W^{\sigma,2}_k(\R^{2d}).
	\end{align*}
	When $\R^{2d}$ is replaced by $\R^d$, as for Lebesgue spaces, we will use shortcut notations and write only for example $H^n$ instead of $H^n(\R^d)$, and $C^\infty_c$ for the space of smooth compactly supported function on $\R^d$.
	
	\subsubsection{Wigner and Weyl transforms} We can associate to each density operator $\op$ a function of the phase space called the Wigner transform and which is defined (for $h=1$) by
	\begin{align*}
		\w(\op)(x,\xi) &:= \intd e^{-2i\pi y\cdot\xi}\r\!\left(x+\frac{y}{2},x-\frac{y}{2}\right)\d y = \F{\tilde{\r}_x}(\xi),
	\end{align*}
	where $\tilde{\r}_x(y) = \r(x+y/2,x-y/2)$ and we used the following convention for the Fourier transform
	\begin{align*}
		\F{u}(\xi) := \intd e^{-2i\pi x\cdot\xi}u(x)\dd x.
	\end{align*}
	This function of the phase space is however not a probability distribution since it is generally not non-negative. We refer to \cite{lions_sur_1993} for more properties of the Wigner transform. Given $\op$, we will write its semiclassical Wigner transform
	\begin{equation*}
		\wh(\op)(x,\xi) := \frac{1}{h^d}\w(\op)\!\(x,\frac{\xi}{h}\).
	\end{equation*}
	Conversely, to each function of the phase space, we can associate an operator through the Weyl transformation, which is the inverse of the Wigner transform. It is defined as the operator such that for any $\varphi\in C^\infty_c$
	\begin{equation*}
		\Weylh{g}\varphi := \iintd g\!\(\tfrac{x+y}{2},\xi\)e^{-i(y-x)\cdot\xi/\hbar}\varphi(y)\dd y\dd\xi.
	\end{equation*}

	\subsubsection{Theorems} Our main result is the following.
	\begin{thm}\label{thm:CV_Hartree}
		Let $d\in\{2,3\}$, $a\in \lt(\max\{\tfrac{d}{2}-2,-1\},d-2\rt]$ and $K$ be given by one of the following expressions
		\begin{equation}\label{eq:condition_K}
			K(x) = \frac{\pm 1}{\n{x}^a} \ \text{ or }\  K(x) = \pm\ln(\n{x}).
		\end{equation}
		In the second case we set $a:=0$. Let $f\geq 0$ be a solution of the Vlasov equation~\eqref{eq:Vlasov} and $\op\geq 0$ be a solution of the Hartree equation~\eqref{eq:Hartree} with respective initial conditions
		\begin{align}\label{eq:initial_condition_f}
			f^\init&\in  W^{\sigma+1,\infty}_m(\R^{2d})\cap H^{\sigma+1}_\sigma(\R^{2d})
			\\\label{eq:initial_condition_op}
			\op^\init&\in \L^1,
		\end{align}
		where $(m,\sigma)\in(4\N)\times(2\N)$ verify $m>d$ and $\sigma>m+\frac{d}{\fb-1}$ with $\fb = \frac{d}{a+1}$. If $a\leq 0$, we also require $\tr((\n{x}^2-\hbar^2\Delta)\,\op^\init)$ to be bounded. Then, there exists $\lambda_f(t)\in C^0(\R_+,\R_+)$ and $C_f(t)\in C^0(\R_+,\R_+)$ depending only on $d$, $a$ and on the initial condition of the solution of the Vlasov equation such that
		\begin{equation}\label{eq:CV_Hartree}
			\Tr{\n{\op - \op_f}} \leq \(\Tr{\n{\op^\init - \op_f^\init}} + C_f(t)\,\hbar \) e^{\lambda_f(t)}
		\end{equation}
		where $\op_f = \Weylh{f}$ and $\op_f^\init = \op_{f^\init}$. An upper bound for the functions $\lambda_f$ and $C_f$ is given by
		\begin{align*}
			\lambda_f(t) &\leq C_{d,a} \int_0^t \Nrm{\Dv f}{W^{n_0,\infty}(\R^{2d})\cap H^\sigma_\sigma(\R^{2d})} \d s
			\\
			C_f(t) &\leq C_{d,a} \int_0^t \Nrm{\rho_f(s)}{L^1\cap H^\nu} \Nrm{\nabla_\xi^2 f(s)}{H^{m}_{m}(\R^{2d})} e^{-\lambda_f(s)}\d s,
		\end{align*}
		which remain bounded at any time $t\geq 0$, and where $\nu = \(\tfrac{m}{2}+a+2-d\)_+$ and $n_0 = \floor{d/2} + 1$.
	\end{thm}
	
	\begin{remark}
		Condition~\eqref{eq:condition_K} includes in particular the Coulomb or Newton potential in dimensions $d=3$ and $d=2$. In these cases, the conditions of regularity~\eqref{eq:initial_condition_f} of the initial data of the Vlasov  equation become $f^\init\in W^{13,\infty}_4(\R^{2d})\cap H^{13}_{12}(\R^{2d})$ when $d=3$ and $a=1$, and $f^\init\in W^{9,\infty}_4(\R^{2d})\cap H^{9}_{8}(\R^{2d})$ when $d=2$ and $a=0$. These conditions are of course not optimal: for example the fact that we ask for $m/2$ and $\sigma$ to be even numbers is mostly to simplify some computations.
	\end{remark}
	
	\begin{remark}
		To see more explicitly that Inequality~\eqref{eq:CV_Hartree} gives a good semiclassical approximation estimate, one can take $\op^\init$ and $\op_f^\init$ such that $\Tr{\n{\op^\init - \op_f^\init}} \leq C \hbar$ and fix some $T>0$, which yields the existence of a constant $C_T>0$ such that for any $t\in[0,T]$
		\begin{equation}\label{eq:CV_Hartree_simplified}
			\Tr{\n{\op - \op_f}} \lesssim C_T\,\hbar.
		\end{equation}
		The theorem also implies the convergence of the spatial density of particles $\rho \to \rho_f$ in $L^1$. Indeed, by duality we have
		\begin{equation}\label{eq:L1-trace}
			\Nrm{\rho-\rho_f}{L^1}=\sup_{\substack{O\in L^\infty(\R^d)\\ \Nrm{O}{L^\infty}\leq 1}}\left|\int O(x)\(\rho(x)-\rho_f(x)\)\d x\right|\leq \Tr{\n{\op-\op_f}}\,,
		\end{equation}
		since every bounded function $x\mapsto O(x)$ also defines a multiplication operator with operator norm $\Nrm{O}{L^\infty}$.
	\end{remark}
	
	From the bound in Theorem \ref{thm:CV_Hartree} we also obtain estimates in other semiclassical Lebesgue spaces.
	\begin{thm}\label{thm:Lp_conv}
		Take the same assumptions and notations as in Theorem~\ref{thm:CV_Hartree}, define $\fb=\frac{d}{a+1}$ and assume moreover that
		\begin{equation*}
			f^\init\in W^{\sigma+1,\infty}_\sigma(\R^{2d})\cap H^{\sigma+1}_\sigma(\R^{2d})
		\end{equation*}
		and that $\sigma > n_0 + \frac{d}{\fb}$. Then for any $p\in[1,\fb)$ it holds
		\begin{align}\label{eq:Lp_conv}
			\Nrm{\op - \op_f}{\L^p} \leq \Nrm{\op^\init - \op_f^\init}{\L^p} + \(\Tr{\n{\op^\init - \op_f^\init}} + c(t) \hbar\) e^{\lambda(t)}\,,
		\end{align}
		where $c$ and $\lambda$ are continuous functions on $\R_+$ depending on $d$, $a$, $p$ and $f^\init$. For any $q\in[\fb,\infty)$, assuming also that $\op^\init\in\L^\infty$, this leads to the following estimate
		\begin{equation}\label{eq:Lp_conv_2}
			\Nrm{\op-\op_f}{\L^q} \leq c_2(t) \(\Nrm{\op^\init-\op_f^\init}{\L^p}^\frac{p}{q} + \Tr{\n{\op^\init-\op_f^\init}}^\frac{p}{q} + \hbar^\frac{p}{q}\) e^{\frac{p}{q} \lambda(t)},
		\end{equation}
		where $\op_f = \Weylh{f}$, $\op_f^\init = \op_{f^\init}$ and $c_2\in C^0(\R_+,\R_+)$ can also be computed explicitly and depends on the initial conditions.
	\end{thm}
		
	\begin{remark}
		In particular, if we assume $\op^\init=\op_f^\init$, or more generally
		\begin{equation*}
			\Tr{\n{\op^\init-\op_f^\init}}\leq C\hbar \quad \text{and} \quad \Nrm{\op^\init-\op_f^\init}{\L^2}\leq C\hbar\,,
		\end{equation*}
		then we have a rate of the form $\hbar^{\fb/2-\eps}$ with $\eps>0$ arbitrarily small. For the Coulomb potential in dimension $d=3$, the estimate reads 
		\begin{equation*}
			\Nrm{f_{\op}-f}{L^2(\R^{2d})} = \Nrm{\op-\op_f}{\L^2}\leq C_T\, \hbar^{3/4-\eps},
		\end{equation*}
		where $f_{\op} = \wh(\op)$ is the Wigner transform of $\op$ and for any $t\in[0,T]$ for some fixed $T>0$. Notice that Theorem~\ref{thm:CV_Hartree} does not imply convergence of the operators, but is only a quantitative estimate, where both $\op$ and $\Weylh{f}$ depend on $\hbar$. As operators, they both for instance converge to $0$ in operator norm since by hypothesis $\Nrm{\op}{\infty} \sim C\,h^d$. On the contrary, the above equation is both a quantitative estimate and a convergence result since $f$ is a fixed element which does not depend on $\hbar$. Thus, it implies the convergence of $f_{\op}$ to $f$ in $L^\infty_\loc(\R_+,L^2(\R^{2d}))$.
		
		With the same assumptions, in the case $d=2$  the Coulomb kernel is of the form $K(x) = C\,\ln(\n{x})$, and $\fb = 2$, implying that Inequality~\eqref{eq:Lp_conv} holds for any $p\in[1,2)$ and that we almost get the conjectured optimal rate of convergence for $p=2$
		\begin{equation*}
			\Nrm{f_{\op}-f}{L^2(\R^{2d})} \leq C_T\, \hbar^{1-\eps}.
		\end{equation*}
	\end{remark}

	Our third result concerns the Hartree-Fock equation. We write both cases $p=1$ and $p>1$ in one theorem in this case. 
	
	\begin{thm}\label{thm:CV_Hartree-Fock}
		Let $\op$ be a solution of the Hartree-Fock~equation~\eqref{eq:HF} and $f$ be a solution of the Vlasov equation~\eqref{eq:Vlasov} satisfying the same initial conditions as in Theorem~\ref{thm:CV_Hartree}, and as in Theorem~\ref{thm:Lp_conv} if $p>1$. If $a> 0$, we assume additionally that the solution has finite kinetic energy, i.e.
		\begin{align*}
			-\Tr{\hbar^2\Delta\op^\init}
		\end{align*}
		is bounded uniformly with respect to $\hbar$. Then, for any $p\in[1,\fb)$, there exist functions $c\in C^0(\R_+,\R_+)$ and $\lambda\in C^0(\R_+,\R_+)$ depending on the $d$, $a$, $p$ and $f^\init$ such that
		\begin{equation*}
			\Nrm{\op - \op_f}{\L^p} \leq \Nrm{\op^\init - \op_f^\init}{\L^p} + \(\Tr{\n{\op^\init - \op_f^\init}} + c(t)\,\hbar^{\min\{1,\tilde{s}-1\}}\) e^{\lambda(t)}\,,
		\end{equation*}
		where $\op_f = \Weylh{f}$, $\op_f^\init = \op_{f^\init}$, $\tilde{s} = d- a_+ - d\(\frac{1}{2} - \frac{1}{p}\)_+$. For $q\in[\fb,\infty)$, assuming again also that $\op^\init\in\L^\infty$, we can still get the following estimate
		\begin{equation*}
			\Nrm{\op - \op_f}{\L^q} \leq c_2(t) \(\Nrm{\op^\init - \op_f^\init}{\L^p}^\frac{p}{q} + \Tr{\n{\op^\init-\op_f^\init}}^\frac{p}{q} + \hbar^{\frac{p}{q}\,\min\{1,\tilde{s}-1\}}\) e^{\frac{p}{q}\lambda(t)}\,.
		\end{equation*}
		where $c_2(t)$ can also be computed explicitly and depends on the initial conditions.
	\end{thm}

\subsection{Discussion}\label{subsec:discussion}

	\subsubsection{Higher singularities} For $a>d-2$, we have no propagation of regularity and therefore our results hold true only in a conditional form. Namely, if the solution to the Vlasov equation is sufficiently regular, then the bounds of Theorem~\ref{thm:CV_Hartree} and Theorem~\ref{thm:Lp_conv} are still satisfied. More precisely, if $d=3$, such conditional results hold for any $a\in(1,2)$. As for Theorem~\ref{thm:CV_Hartree-Fock}, a conditional result is still true. However, due to the need of controlling the exchange term $\sfX$, we can only address a smaller class of potentials. In particular in dimension $d=3$ we have $a\in(1,3/2)$. %In general, when $d>1$, then $a<\frac{2(d-1)}{d+2}$.
	Our results in dimension $2$ and $3$ can be summarized as follows.\\
	\begin{center}
	\begin{tabular}{|rcl|c|c|}
		\hline&&&&\\[-1em]
		\multicolumn{3}{c}{Settings} & Hartree & Hartree-Fock \\
		\hline&&&&\\[-1em]
		$d = 2$ & and & $a\in (-1,0]$ & global & global \\
		&&&&\\[-1em]
		$d = 2$ & and & $a\in \lt(0,1\rt)$ & conditional & conditional \\[0.1em]
		\hline&&&&\\[-1em]
		$d = 3$ & and & $a\in \lt(-\frac{1}{2},1\rt]$ & global & global \\[0.2em]
		&&&&\\[-1em] 
		$d = 3$ & and & $a\in\lt(1,\frac{3}{2}\rt)$ & conditional & conditional \\[0.2em]
		&&&&\\[-1em]
		$d = 3$ & and & $a\in \lt[\frac{3}{2},2\rt)$ & conditional & no results \\[0.2em]
		\hline
	\end{tabular}
	\end{center}~

	\subsubsection{General class of potentials}\label{subsec:generalization_K}
	
		All our results generalize to more general non-radial pair interactions. For $s\in(0,d)$, define the weak Sobolev space $\dot{H}^{s,1}_w$ as the completion of $C^\infty_c$ with respect to the norm
		\begin{align*}
			\Nrm{u}{\dot{H}^{s,1}_w} := \Nrm{\Delta^\frac{s}{2}u}{\mathrm{TV}},
		\end{align*}
		where $\Nrm{\cdot}{\mathrm{TV}}$ denotes the total variation norm over the space $\cM$ of bounded measures. By the formula of the inverse of the powers of Laplacian, we deduce that it is the space of functions that can be written
		\begin{equation}\label{eq:representation_formula}
			u(x) = \intd \frac{1}{\n{x-w}^{d-s}}\,\mu(\d w),
		\end{equation}
		for some measure $\mu\in\cM$.
		
		Notice that this space contains the interaction kernel
		\begin{align*}
			K(x) = \frac{1}{\n{x}^a} \text{ with } a = d-s,
		\end{align*}
		when $a>0$, which follows by taking $\mu = \delta_0$. In particular, for the Coulomb potential in dimension $d=3$, it holds
		\begin{equation*}
			\frac{1}{\n{x}} \in \dot{H}^{2,1}_w.
		\end{equation*}
		However, this space contains also more general potentials. It contains for example the Sobolev space $\dot{H}^{s,1} = \dot{F}^s_{2,1}$ which is defined by the norm $\Nrm{\Delta^\frac{s}{2}u}{L^1}$. When $n\in\N$, then $\dot{H}^{n,1} = \dot{W}^{n,1}$ is a classical homogeneous Sobolev space.
		
		The proof for more general potentials follows mainly from the fact that the equation and most of our estimates depend linearly on $K$. As an example, Proposition~\ref{prop:exchange-term} is proved with this class of potentials. Hence, all our results also hold with the assumption $K\in\dot{H}^{d-a}_w$ instead of $K(x)= \n{x}^{-a}$ when $a>0$, except Theorem~\ref{thm:CV_Hartree-Fock}, since we need an assumption on $K^2$ to prove inequalities~\eqref{eq:bound_X_p_small} and~\eqref{eq:bound_X_p_big}. For this theorem, the assumption $K(x)= \n{x}^{-a}$ can therefore be replaced by $K\in\dot{H}^{d-a}_w$ and $K^2\in\dot{H}^{d-2a}_w$ when $a\geq 0$.
	
	\subsubsection{From Hartree to Hartree-Fock} Notice that Theorem~\ref{thm:Lp_conv} and Theorem~\ref{thm:CV_Hartree-Fock} give a semiclassical estimate between the solutions of the Hartree equation~\eqref{eq:Hartree} and the solutions of the Hartree-Fock equation~\eqref{eq:HF}. Indeed, let $\op_\text{H}$ and $\op_{\text{HF}}$ be respectively solutions to the Hartree equation and the Hartree-Fock equation, and let $\op_f$ be a solution to the Weyl transformed Vlasov equation. Then, for $p\in[1,\infty)$, we have 
	\begin{equation*}
		\Nrm{\op_H-\op_{\text{HF}}}{\L^p} \leq \Nrm{\op_H-\op_f}{\L^p} + \Nrm{\op_{\text{HF}}-\op_f}{\L^p},
	\end{equation*}
	where the first term in the r.h.s. is bounded by Theorem~\ref{thm:Lp_conv} and the second term in the r.h.s. can be estimated by Theorem~\ref{thm:CV_Hartree-Fock}.
	
	\subsubsection{Well-posedness} One of the strength of the method is that our strong regularity assumptions that have to be independent of $\hbar$ only concerns the solutions of the Vlasov equation. Our assumptions on the solution of the Hartree equation imply the global well-posedness of solutions as proved in \cite{castella_l2_1997} where the trace norm corresponds to the $L^2(\lambda)$ norm (see also \cite{ginibre_class_1980, ginibre_global_1985, lions_sur_1993}). Even if these assumptions are weak, observe however that the operator $\op^\init$ has to be at a finite trace norm distance from the operator $\op^\init$ which by construction is bounded in higher Sobolev spaces (as can be deduced from Proposition~\ref{prop:multiply_Weyl_L2}). The additional moment bound when $a\leq 0$ ensures that the energy is finite, which allows to propagate the space moments (see e.g. \cite[Remark~3.1]{lafleche_propagation_2019}). This is sufficient to give a meaning to the pair interaction potential which is growing at infinity in this case.

\section{Strategy}\label{sec:strategy}

	The strategy of this paper consists in getting the semiclassical analogue of the estimates of  classical mechanics, and in particular the case of kinetic models. The quantum analogue of the classical momentum variable $\xi$ is the operator
	\begin{equation*}
		\opp := -i\hbar\nabla,
	\end{equation*}
	which is an unbounded operator on $L^2$. From this we get in particular that $|\opp|^2 := \opp^*\opp = -\hbar^2\Delta$ and we can express the Hamiltonian~\eqref{eq:Hartree-Ham} as $H = \frac{\n{\opp}^2}{2} + V(x)$.

\subsection{Quantum gradients of the phase space}\label{subsec:gradients}
	
	Since our method here uses regular initial conditions, we define the following operators which are the quantum equivalent of the gradient with respect to the variables $x$ and $\xi$ of the phase space:
	\begin{align*}
		\Dhx\op &:= [\nabla,\op] = \com{\frac{\opp}{i\hbar},\op}
		\\
		\Dhv\op &:= \com{\frac{x}{i\hbar},\op}.
	\end{align*}
	These formulas can be seen from the point of view of the correspondence principle as the quantum equivalent of the Poisson bracket definition of the classical gradients. Another point of view is to observe that they are Weyl quantizations, since we have
	\begin{align*}
		\Dhx\op = \Weylh{\Dx \wh(\op)},
		\\
		\Dhv\op = \Weylh{\Dv \wh(\op)}.
	\end{align*}
	One should not confuse $\nabla\in\L(L^2)$ with $\Dhx \in \L(\L(L^2))$. In Section~\ref{sec:regularity}, we prove that if a function on the phase space is sufficiently smooth in the classical sense, then its Weyl quantization also has some smoothness in the semiclassical sense.

\subsection{The classical case: \texorpdfstring{$L^1$}{L1} weak-strong stability}\label{subsec:classical_case}

	In the classical case, the method we use to prove the semiclassical limit, which is the content of Section~\ref{sec:proof_thm_1_2} and Section~\ref{sec:proof_thm_3} can be seen as an equivalent of the following $L^1$ weak-strong stability estimate for the Vlasov equation, which tells that we just need to have a control of the gradient of only one of the solutions to get a bound on the integral of their difference.
	
	For functions on the phase space of the form $f=f(x,\xi)$, we use the shortcut notation $L^p_xL^{q,r}_\xi = L^p(\R^d,L^{q,r}(\R^d))$. The next proposition can be seen as the classical equivalent of Theorem~\ref{thm:CV_Hartree}.
	
	\begin{prop}\label{prop:classical_case}
		Let $\fb\in(1,\infty]$ and $\nabla K\in L^{\fb,\infty}$ and assume $f_1$ and $f_2$ are two solutions of the Vlasov equation~\eqref{eq:Vlasov} in $L^\infty([0,T],L^1(\R^{2d}))$ for some $T>0$. Then, under the condition
		\begin{equation}\label{eq:classical_regu_condition}
			\nabla_\xi f_2 \in L^1([0,T],L^{\fb',1}_xL^1_{\xi}),
		\end{equation}
		one has the following stability estimate
		\begin{align*}
			\Nrm{f_1-f_2}{L^1(\R^{2d})} \leq \Nrm{f_1^\init-f_2^\init}{L^1(\R^{2d})} \exp\(C \int_0^T \Nrm{\Dv f_2}{L^{\fb',1}_xL^1_{\xi}}\d t\),
		\end{align*}
		where $C = \Nrm{\nabla K}{L^{\fb,\infty}}$.
	\end{prop}
	
	\begin{remark}
		In the case of the Coulomb interaction, $\fb = \frac{3}{2}$, the condition on $f_2$ becomes
		\begin{equation*}
			\intd\n{\nabla_\xi f_2}\d\xi \in L^1([0,T],L^{3,1}_x),
		\end{equation*}
		which by real interpolation follows in particular if 
		\begin{equation*}
			\Nrm{\nabla_\xi f_2}{L^1_\xi} \in L^1([0,T],L^{3+\eps}_x\cap L^{3-\eps}_x),
		\end{equation*}
		for some $\eps\in(0,2]$. In particular, the case $\eps = 2$ yields $(3-\eps,3+\eps) = (1,5)$, which corresponds to the equivalent of the hypotheses required on the solutions in \cite{saffirio_hartree_2020}. A quantum version of this hypothesis can also be found in \cite{porta_mean_2017}. 
	\end{remark}
	
	\begin{remark}
		This result allows $\nabla K$ to be more singular than the case of the Coulomb potential. However, it is a conditional result, since one still has to show that condition~\eqref{eq:classical_regu_condition} holds. If the potential is the Coulomb potential or a less singular potential, then one can prove that this condition holds if the data is initially in some weighted Sobolev space by Proposition~\ref{prop:regu_Vlasov} in Appendix~\ref{sec:appendix_A}. If the potential is more singular than the Coulomb potential, then it is not clear that there are cases such that condition~\eqref{eq:classical_regu_condition} is verified.
	\end{remark}

	\begin{proof}[Proof of Proposition~\ref{prop:classical_case}]
		Let $f:= f_1-f_2$ and define for $k\in\{1,2\}$, $\rho_k = \intd f_k\dd\xi$ and $E_k = \nabla V_k = \nabla K * \rho_k$. Then it holds
		\begin{align*}
			\dpt f + \xi\cdot\Dx f + E_1\cdot\Dv f = (E_2-E_1)\cdot\Dv f_2,
		\end{align*}
		so that by defining $\rho := \rho_1-\rho_2$, we obtain
		\begin{align*}
			\dpt \iintd \n{f}\d x\dd \xi &= \iintd \(\nabla K * \rho \cdot \Dv f_2\) \Sign{f} \d x\dd\xi
			\\
			&= -\intd \rho\, \nabla K\,\dot{*} \(\intd \Sign{f} \Dv f_2\dd\xi\)
			\\
			&\leq \Nrm{f}{L^1} \Nrm{\nabla K * \intd \n{\Dv f_2}\d\xi}{L^\infty},
		\end{align*}
		where the notation $\dot{*}$ indicates that we perform the dot product of vectors inside the convolution. We conclude by noticing that by H\"older's inequality for Lorentz spaces (see for example~\cite[Formula~(2.7)]{hunt_lp_1966}), for any $g\in L^{\fb',1}$, the following inequality holds
		\begin{equation}\label{eq:Holder_lorentz}
			\Nrm{\nabla K * g}{L^\infty}\leq \sup_{z\in\R^d} \intd \n{\nabla K(z-\cdot)g} \leq \Nrm{\nabla K}{L^{\fb,\infty}} \Nrm{g}{L^{\fb',1}},
		\end{equation}
		so that the result follows by taking $g = \Nrm{\nabla_\xi{f_2}}{L^1_\xi}$ and then using Gr\"{o}nwall's Lemma.
	\end{proof}
	
	The next proposition is the classical equivalent of the first part of Theorem~\ref{thm:Lp_conv}.
	
	\begin{prop}\label{prop:classical-Lp}
		Let $\fb> 1$ and $\nabla K\in L^{\fb,\infty}$ and assume $f_1$ and $f_2$ are two solutions of the Vlasov equation~\eqref{eq:Vlasov} in $L^\infty([0,T],L^1(\R^{2d}))$ for some $T>0$. Then, if $\nabla_\xi f_2 \in L^1([0,T],L^{q,1}_xL^p_{\xi})$, the following inequality holds
		\begin{align*}
			\Nrm{f_1-f_2}{L^p(\R^{2d})} \leq \Nrm{f_1^\init-f_2^\init}{L^p(\R^{2d})} \exp\(C \int_0^T \Nrm{\Dv f_2}{L^{q,1}_xL^p_{\xi}}\d t\),
		\end{align*}
		where $C = \Nrm{\nabla K}{L^{\fb,\infty}}$ and 
		\begin{equation}\label{eq:q_to_p}
			\frac{1}{q}=\frac{1}{p}-\frac{1}{\fb}.
		\end{equation}
	\end{prop}
	
	\begin{remark}
		Formula~\eqref{eq:q_to_p} implies $p\leq\fb$. In the case of the Coulomb interaction in dimension $d=3$ we have $\fb = \frac{3}{2}$, thus the estimate works at most with $p=\frac{3}{2}$.
	\end{remark}

	\begin{proof}	
		We define the two parameters semigroup $S_{t,s}$ such that $S_{s,s} = \Id$ and 
		\begin{align*}
			\dpt\, S_{t,s} g = \Lambda_t S_{t,s} g\,,
		\end{align*}
		where 
		\begin{equation*}
			\Lambda_t S_{t,s} g:=-\xi\cdot\nabla_x S_{t,s} g - E_1(t)\cdot\nabla_\xi S_{t,s} g\,,
		\end{equation*}
		with $E_1(t)=E_1(t,x)=-\nabla K*\rho_1$ with $\rho_1(t,x)=\int f_1(t,x,\xi)\dd \xi$. Now observe that the flow property of $S_{t,s}$ implies that $\dps S_{t,s} = - S_{t,s}\Lambda_s$. Thus, using the notation
		\begin{equation*}
			\tilde{\Lambda}_t := -\xi\cdot\Dx - E_2(t)\cdot\Dv,
		\end{equation*}
		and taking $f_1(s) = f_1(s,x,\xi)$ and $f_2(s) = f_2(s,x,\xi)$ two solutions of the Vlasov equation, we get
		\begin{align*}
			\dps S_{t,s} (f_1-f_2)(s) &= - S_{t,s} \Lambda_s (f_1-f_2)(s) + S_{t,s}\Lambda_s f_1(s) - S_{t,s}\tilde{\Lambda}_s f_2(s)
			\\
			&= S_{t,s}\(\Lambda_s-\tilde{\Lambda}_s\) f_2(s)
			\\
			&= S_{t,s}\(\(E_2(s)-E_1(s)\)\cdot \Dv f_2(s)\),
		\end{align*}
		and by integrating with respect to $s$ and denoting $f := f_1-f_2$ and $E:= E_1-E_2$, we obtain the following Duhamel formula
		\begin{equation*}
			f(t) = S_{t,0} f^\init + \int_0^t S_{t,s}\(E(s)\cdot \Dv f_2(s)\) \dd s.
		\end{equation*}
		Since the semigroup $S_{t,s}$ preserves all Lebesgue norms of the phase space, taking the $L^p$ norm yields
		\begin{align*}
			\Nrm{f(t)}{L^p_{x,\xi}} \leq \Nrm{f^\init}{L^p_{x,\xi}} + \int_0^t \Nrm{E(s)\cdot \Dv f_2(s)}{L^p_{x,\xi}} \d s.
		\end{align*}
		To bound the expression inside the time integral we write
		\begin{align*}
			\Nrm{E(s)\cdot \Dv f_2(s)}{L^p_{x,\xi}} &= \Nrm{(\rho*\nabla K)\cdot \Dv f_2(s)}{L^p_{x,\xi}}
			\\
			&\leq \intd \n{\rho(z)} \Nrm{\nabla K(\cdot-z)\cdot \Dv f_2(s)}{L^p_{x,\xi}}\d z
			\\
			&\leq \intd \n{\rho(z)} \Nrm{\n{\nabla K(\cdot-z)}\Nrm{\Dv f_2(s)}{L^p_v}}{L^p_x}\d z
			\\
			&\leq \Nrm{\rho}{L^1} \Nrm{\nabla K}{L^{\fb,\infty}}\Nrm{\Dv f_2(s)}{L^{q,1}_{x} L^p_v},
		\end{align*}
		where we used again H\"{o}lder's inequality for Lorentz spaces.
	\end{proof}

\section{Regularity of the Weyl transform}\label{sec:regularity}

	In this section, we prove that if the solution $f$ of the Vlasov equation is sufficiently well-behaved, then we can obtain uniform in $\hbar$ bounds for the quantum equivalent of the norm $\Nrm{\nabla_\xi f}{L^p_xL^1_\xi}$ expressed in term of the Weyl transform of $f$.

	\begin{prop}\label{prop:diag_Dhv_Weyl}
		Let $(n,n_1)\in \N^2$ be even numbers such that $n > d/2$ and define $\sigma:=2n+n_1$ and $n_0 = \floor{d/2}+1$. Then, for any $f\in W^{n_0+1,\infty}(\R^{2d})\cap H^{\sigma+1}_\sigma(\R^{2d})$, there exists a constant $C_{d,n_1}>0$ depending only on $d$ and $n_1$ such that
		\begin{align*}
			\Nrm{\Diag{|\Dhv{\Weylh{f}}|}}{L^p} \leq C_{d,n_1}\,\Nrm{\Dv f}{W^{n_0,\infty}(\R^{2d})\cap H^\sigma_\sigma(\R^{2d})}
		\end{align*}
		for any $p\in [1,1+\frac{n_1}{d}]$.
	\end{prop}

	The strategy is to use a special case of the quantum kinetic interpolation inequality proved in \cite[Theorem~6]{lafleche_propagation_2019}. For the operator $|\Dhv{\op}|$, it reads
	\begin{equation}\label{eq:interpolation}
		\Nrm{\Diag{|\Dhv{\op}|}}{L^p} \leq C\(\Tr{\n{\Dhv{\op}}\n{\opp}^{n_1}}\)^{\theta} \nrm{\Dhv{\op}}_{\L^\infty}^{1-\theta},
	\end{equation}
	where $p$ is given by $p = 1 + \frac{n_1}{d}$ and $\theta = \frac{1}{p}$. The corresponding kinetic inequality is
	\begin{equation*}
		\Nrm{\Dv f}{L^p_x(L^1_\xi)} \leq C \(\iintd\n{\Dv f}\n{\xi}^{n_1}\d x\dd\xi\)^{\theta} \nrm{\Dv f}_{L^\infty_{x,\xi}}^{1-\theta}.
	\end{equation*}
	To do that, we will need to compare the multiplication by $\n{\opp}^n$ and $\n{x}^n$ of the Weyl transform of a phase space function $g$ with the Weyl transform of the multiplication of $g$ with $\n{\opp}^n$ and $\n{x}^n$. This makes appear error terms involving derivatives of $g$. For example, in the case $n=2$, it holds
	\begin{align*}
		\Weylh{g}\n{\opp}^2 &= \Weylh{\n{\xi}^2 g + \frac{i\hbar}{2}\,\xi\cdot\Dx g - \frac{\hbar^2}{4}\,\Delta_x g}
		\\
		\Weylh{g}\n{x}^2 &= \Weylh{\n{x}^2 g + i\hbar\,\xi\cdot\Dv g + \frac{\hbar^2}{4} \Delta_\xi g}.
	\end{align*}
	More generally, one can obtain similar identities when $n\in\N$. In order to write them, we introduce the standard multi-index notations
	\begin{align*}
		\alpha &:= \(\alpha_\ii\)_{\ii\in\Int{1,d}} \in \N^d,
		\\
		\n{\alpha} &:= \sum_{\ii=1}^{d} \alpha_\ii & \alpha! &:= \alpha_1!\,\alpha_2!\, \dots \alpha_d! 
		\\
		x^\alpha &:= x_1^{\alpha_1} x_2^{\alpha_2}\dots x_d^{\alpha_d} &\partial_x^\alpha &:= \partial_{x_1}^{\alpha_1}\partial_{x_2}^{\alpha_2}\dots \partial_{x_d}^{\alpha_d}
		\\
		\alpha \leq \beta &\ssi \forall\ii\in\Int{1,d}, \alpha_\ii\leq \beta_\ii.
	\end{align*}
	We then obtain the following set of identities.
	\begin{lem}\label{lem:multiply_Weyl}
		For any $n\in 2\N$ and any tempered distribution $g$ of the phase space, it holds
		\begin{subequations}
		\begin{align}\label{eq:Weyl_vs_p}
			\Weylh{g}\n{\opp}^n &= \sum_{|\alpha+\beta| = n} a_{\alpha,\beta} \(\tfrac{i\hbar}{2}\)^{\n{\beta}} \Weylh{\xi^\alpha\,\partial_x^\beta g}
			\\\label{eq:Weyl_vs_x}
			\Weylh{g}\n{x}^{n} &= \sum_{|\alpha+\beta| = n} b_{\alpha,\beta}  \(-i\hbar\)^{\n{\beta}} \Weylh{x^{\alpha}\,\partial_\xi^\beta g}
			\\\label{eq:Weyl_vs_xp}
			\Weylh{g}\n{\opp}^{n_1}\n{x}^{n} &= \sum_{\substack{|\alpha+\beta| = n_1\\ |\alpha'+\beta'| = n}} a_{\alpha,\beta} b_{\alpha',\beta'} \(-i\hbar\)^{\n{\beta'}}\(\tfrac{i\hbar}{2}\)^{\n{\beta}} \Weylh{x^{\alpha'}\partial_\xi^{\beta'}\!\!\(\xi^\alpha\,\partial_x^\beta g\)}.
		\end{align}
		\end{subequations}
		where the coefficients $a_{\alpha,\beta}$, $b_{\alpha,\beta}$ and $c = c_{\alpha,\beta,\alpha',\beta',\gamma}$ are nonnegative coefficients that do not depend on $\hbar$.
	\end{lem}

	\begin{proof}[Proof of Lemma~\ref{lem:multiply_Weyl}]
		By definition of the Weyl transform, we deduce that for any $\varphi\in C^\infty_c$ it holds
		\begin{align*}
			\Weylh{g}\n{\opp}^n\varphi &= \(i\hbar\)^n \iintd g\!\(\tfrac{x+y}{2},\xi\)e^{-i(y-x)\cdot\xi/\hbar} \Delta^\frac{n}{2}\varphi(y)\dd y\dd\xi
			\\
			&= \(i\hbar\)^n \iintd \Delta_y^\frac{n}{2}\!\(g\!\(\tfrac{x+y}{2},\xi\)e^{-i(y-x)\cdot\xi/\hbar}\) \varphi(y)\dd y\dd\xi.
		\end{align*}
		With the multi-index notation, we can expand the powers of the Laplacian of a product of functions in the following way
		\begin{align*}
			\Delta^\frac{n}{2}(fg) = \sum_{|\alpha+\beta| = n} a_{\alpha,\beta}\, \partial^\alpha f\, \partial^\beta g,
		\end{align*}
		where the $a_{\alpha,\beta}^n$ are nonnegative constants depending on $n$ and on the multi-index $\alpha$, and such that
		\begin{equation*}
			\sum_{|\alpha+\beta| = n} a_{\alpha,\beta} = (4d)^n.
		\end{equation*}
		Thus, we deduce that the integral kernel $\kappa$ of the operator $\Weylh{g}\n{\opp}^n$ is given by
		\begin{align*}
			\kappa(x,y) &= \sum_{|\alpha+\beta| = n} a_{\alpha,\beta} \(i\hbar\)^{n-\n{\alpha}} \intd 2^{-\n{\beta}} \partial_x^\beta g\!\(\tfrac{x+y}{2},\xi\) \xi^\alpha e^{-i(y-x)\cdot\xi/\hbar} \dd\xi,
		\end{align*}
		which yields, 
		\begin{equation*}
			\Weylh{g}\n{\opp}^n = \sum_{|\alpha+\beta| = n} a_{\alpha,\beta} \(\tfrac{i\hbar}{2}\)^{\n{\beta}} \Weylh{\xi^\alpha\,\partial_x^\beta g}.
		\end{equation*}
		This proves Identity~\eqref{eq:Weyl_vs_p}. To prove the second identity, we write $u := \frac{x+y}{2}$ and $v := y-x$ so that the integral kernel $\kappa_2$ of the operator $\Weylh{g}\n{x}^2$ is given by
		\begin{align*}
			\kappa_2(x,y) &= \iintd g\!\(\tfrac{x+y}{2},\xi\)e^{-i\(y-x\)\cdot\xi/\hbar} \n{y}^n \d\xi
			\\
			&= \iintd g\!\(u,\xi\)e^{-i\,v\cdot\xi/\hbar} \(\n{u+\frac{v}{2}}^2\)^\frac{n}{2} \d\xi
			\\
			&= \iintd g\!\(u,\xi\)e^{-i\,v\cdot\xi/\hbar} \(\sum_{\ii=1}^d \( u_\ii^2 + \frac{v_\ii^2}{4} + u_\ii v_\ii\) \)^\frac{n}{2} \d\xi.
		\end{align*}
		By the multinomial theorem, this can be written in the form
		\begin{align*}
			\kappa_2(x,y) &= \sum_{\n{\alpha+\beta} = n} b_{\alpha,\beta} \iintd u^\alpha g\!\(u,\xi\) v^{\beta} e^{-i\,v\cdot\xi/\hbar} \d\xi
			\\
			&= \sum_{\n{\alpha+\beta} = n} b_{\alpha,\beta} \iintd u^\alpha g\!\(u,\xi\) \(i\hbar\)^{\n{\beta}}\partial_\xi^\beta e^{-i\,v\cdot\xi/\hbar} \dd\xi
			\\
			&= \sum_{\n{\alpha+\beta} = n} b_{\alpha,\beta} \(-i\hbar\)^{\n{\beta}} \iintd u^\alpha \partial_\xi^\beta g\!\(u,\xi\) e^{-i\,v\cdot\xi/\hbar} \dd\xi,
		\end{align*}
		where we used $\n{\beta}$ times integration by parts to get the last line, and the $b_{\alpha,\beta}$ are nonnegative constants that satisfy
		\begin{align*}
			\sum_{|\alpha+\beta| = n} b_{\alpha,\beta} = \(\frac{9\,d}{4}\)^\frac{n}{2}.
		\end{align*}
		In term of operators, this yields the following identity 
		\begin{equation*}
			\Weylh{g}\n{x}^n = \sum_{|\alpha+\beta| = n} b_{\alpha,\beta} \(-i\hbar\)^{\n{\beta}} \Weylh{x^\alpha\,\partial_\xi^\beta g}.
		\end{equation*}
		This yields Identity~\eqref{eq:Weyl_vs_x}. To get the last identity, we combine the two first to get
		\begin{align*}
			\Weylh{g}\n{\opp}^{n_1}\n{x}^{n} &= \sum_{|\alpha+\beta| = n_1} a_{\alpha,\beta} \(\tfrac{i\hbar}{2}\)^{\n{\beta}} \Weylh{\xi^\alpha\,\partial_x^\beta g} \n{x}^n
			\\
			&= \sum_{\substack{|\alpha+\beta| = n_1\\ |\alpha'+\beta'| = n}} a_{\alpha,\beta} b_{\alpha',\beta'} \(-i\hbar\)^{\n{\beta'}}\(\tfrac{i\hbar}{2}\)^{\n{\beta}} \Weylh{x^{\alpha'}\partial_\xi^{\beta'}\!\!\(\xi^\alpha\,\partial_x^\beta g\)}.
		\end{align*}
	\end{proof}
	
	From this lemma, we deduce the following $\L^2$ inequalities.
	\begin{prop}\label{prop:multiply_Weyl_L2}
		Let $n\in 2\N$ and $g$ a function of the phase space, then there exists a constant $C>0$ depending only on $d$ and $n$ such that
		\begin{subequations}
		\begin{align}\label{eq:Weyl_vs_p_L2}
			\Nrm{\Weylh{g}\n{\opp}^n}{\L^2} &\leq \(4d\)^n \(\Nrm{g \n{\xi}^n}{L^2(\R^{2d})} + \(\tfrac{\hbar}{2}\)^n \Nrm{\nabla_x^n g}{L^2(\R^{2d})}\)
			\\\label{eq:Weyl_vs_x_L2}
			\Nrm{\Weylh{g}\n{x}^n}{\L^2} &\leq \(\tfrac{9d}{4}\)^n \(\Nrm{g \n{x}^n}{L^2(\R^{2d})} + \hbar^n \Nrm{\nabla_\xi^n g}{L^2(\R^{2d})}\)
			\\\nonumber
			\Nrm{\Weylh{g}\n{\opp}^{n_1}\!\n{x}^n}{\L^2} &\leq C \left(\Nrm{\(1+\n{x}^n\n{\xi}^{n_1}\)g}{L^2(\R^{2d})} + \hbar^{n_1} \Nrm{\n{x}^n \nabla_x^{n_1} g}{L^2(\R^{2d})}\right.
			\\\label{eq:Weyl_vs_xp_L2}
			&\qquad \left. +\, \hbar^n \Nrm{\n{\xi}^{n_1}\nabla_\xi^n g}{L^2(\R^{2d})} + \hbar^{n+n_1} \Nrm{\nabla_x^{n_1}\nabla_\xi^n g}{L^2(\R^{2d})}\right).
		\end{align}
		\end{subequations}
	\end{prop}
	
	\begin{proof}[Proof of Proposition~\ref{prop:multiply_Weyl_L2}]
		By Formula~\eqref{eq:Weyl_vs_p} and the fact that for any $u\in L^2(\R^{2d})$, $\Nrm{\Weylh{u}}{\L^2} = \Nrm{u}{L^2(\R^{2d})}$, we obtain
		\begin{align}
			\Nrm{\Weylh{g}\n{\opp}^n}{\L^2} &\leq \sum_{|\alpha+\beta| = n} a_{\alpha,\beta} \(\tfrac{\hbar}{2}\)^{\n{\beta}} \Nrm{\xi^\alpha\,\partial_x^\beta g}{L^2(\R^{2d})}.
		\end{align}
		Then, for any multi-index $\alpha$ and $\beta$ such that $\n{\alpha+\beta} = n$, by defining $\hat{g}(y,\xi)$ as the Fourier transform of $g(x,\xi)$ with respect to the variable $x$, the fact that the Fourier transform is unitary in $L^2_x$ yields
		\begin{align*}
			\(\tfrac{\hbar}{2}\)^{\n{\beta}}\Nrm{\xi^\alpha\,\partial_x^\beta g}{L^2(\R^{2d})} &= \(\tfrac{h}{2}\)^{\n{\beta}} \Nrm{\xi^\alpha\,y^\beta \,\hat{g}}{L^2(\R^{2d})}
			\\
			&\leq \tfrac{\n{\alpha}}{n}\Nrm{\n{\xi}^n \hat{g}}{L^2(\R^{2d})} + \tfrac{\n{\beta}}{n}\, \(\tfrac{h}{2}\)^{\n{\beta}}\Nrm{\n{y}^n g}{L^2(\R^{2d})}
			\\
			&\leq \Nrm{\n{\xi}^n g}{L^2(\R^{2d})} + \(\tfrac{\hbar}{2}\)^n \Nrm{\nabla_x^n \hat{g}}{L^2(\R^{2d})}.
		\end{align*}
		Moreover, as remarked in the proof of Lemma~\ref{lem:multiply_Weyl}, it holds
		\begin{equation*}
			\sum_{|\alpha+\beta| = n} a_{\alpha,\beta} = (4d)^n,
		\end{equation*}
		from which we obtain Inequality~\eqref{eq:Weyl_vs_p_L2}. Formulas~\eqref{eq:Weyl_vs_x_L2} and \eqref{eq:Weyl_vs_xp_L2} can be proved in the same way.
	\end{proof}
	
	Moreover, we can bound weighted $\L^1$ norms using $\L^2$ norms with bigger weights. This is the content of the following proposition where we recall the notation $\weight{y} = \sqrt{1+\n{y}^2}$ for the weights.
	\begin{prop}\label{prop:moment_bound}
		Let $(n,n_1)\in\N^2$ be even numbers such that $n>d/2$ and define $k:=n+n_1$. Assume $\op := \Weylh{g}$ is the Weyl transform of a function $g\in H^{n+k}_{n+k}(\R^{2d})$. Then the following inequality holds
		\begin{align*}
			\Tr{\n{\op}\n{\opp}^{n_1}} &\leq C \left(\Nrm{\weight{\xi}^k\weight{x}^n g}{L^2(\R^{2d})} + \hbar^k \Nrm{\weight{x}^n\Dx^k g}{L^2(\R^{2d})}\right.
			\\
			&\qquad\quad \left. + \hbar^n \Nrm{\weight{\xi}^k\Dv^n g}{L^2(\R^{2d})} + \hbar^{k+n} \Nrm{\Dx^k\Dv^n g}{L^2(\R^{2d})}\right).
		\end{align*}
	\end{prop}
	
	\begin{proof}
		First notice that since the sum of eigenvalues is always smaller than the sum of singular values (see for example Formula~(3.1) in \cite{simon_trace_2005}), it holds
		\begin{align*}
			\Tr{\n{\op}\n{\opp}^{n_1}} \leq \Tr{\n{\big(\n{\op}\n{\opp}^{n_1}\big)}},
		\end{align*}
		and from the definition of $\n{AB}$ if $A$ and $B$ are two operators, we see that $\n{AB} = (B^*A^*AB)^\frac{1}{2} = \n{\n{A}B}$, so that $\Tr{\n{\big(\n{\op}\n{\opp}^{n_1}\big)}} = \Tr{\n{\op\n{\opp}^{n_1}}}$. Defining $\opm_n := \(1+|\opp|^n\)\(1+|x|^n\)$, we deduce from the Cauchy-Schwarz inequality that
		\begin{equation}\label{eq:L2m_to_L1_CS}
			\Tr{\n{\op}\n{\opp}^{n_1}} \leq \Tr{\n{\op\n{\opp}^{n_1}}} \leq \Nrm{\op\n{\opp}^{n_1}\opm_n}{2} \Nrm{\opm_n^{-1}}{2}.
		\end{equation}
		To control the second factor in the right-hand side, we observe that it is of the form $\opm_n^{-1} = w(x)\,w(-i\hbar\nabla)$ with $w(y) = \(1+\n{y}^n\)^{-1}$, so that its Hilbert-Schmidt norm can be computed exactly (see e.g. \cite[Equation~(4.7)]{simon_trace_2005})
		\begin{equation*}
			\Nrm{\opm_n^{-1}}{2} = (2\pi)^{-d/2} \Nrm{w}{L^2} \Nrm{w(\hbar\,\cdot)}{L^2}  = C_{d,n}\, h^{-d/2},
		\end{equation*}
		where $C_{d,n} = \Nrm{w}{L^2}^2$ is finite since $n>d/2$. Therefore, by definition of the $\L^2$ norm, Inequality~\eqref{eq:L2m_to_L1_CS} leads to
		\begin{align*}
			\Tr{\n{\op}\n{\opp}^{n_1}} &\leq C_{d,n}\Nrm{\op\n{\opp}^{n_1}\opm_n}{\L^2}
			\\
			&\leq C_{d,n} \(\Nrm{\op\(\n{\opp}^{n_1}+ \n{\opp}^{n_1}\n{x}^n + \n{\opp}^{n+n_1} + \n{\opp}^{n+n_1}\n{x}^n\)}{\L^2}\).
		\end{align*}
		To get the result, we take $\op = \Weylh{g}$ and use Proposition~\ref{prop:multiply_Weyl_L2} to bound the right-hand side of the above inequality by weighted classical $L^2$ norms of $g$.
	\end{proof}

	We can now prove the main proposition of this section following the strategy explained at the beginning of this section.
	
	\begin{proof}[Proof of Proposition~\ref{prop:diag_Dhv_Weyl}]
		We use an improvement of the Calder\'{o}n-Vaillancourt theorem for Weyl operators proved by Boulkhemair in \cite{boulkhemair_l2_1999} which states that if $g\in W^{n_0,\infty}(\R^{2d})$ with $n_0 = \floor{\frac{d}{2}} +1$, then $\Weyl{g}$ is a bounded operator on $L^2$ and its operator norm is bounded by
		\begin{equation}\label{eq:Calderon_Vaillancourt}
			\Nrm{\Weyl{g}}{\B(L^2)} \leq C \nrm{g}_{W^{n_0,\infty}(\R^{2d})}.
		\end{equation}
		Since $\Weylh{g} = h^d\, \Weyl{g(\,\cdot,h\,\cdot)}$, and that for $h\leq 1$
		\begin{equation*}
			\Nrm{g(\,\cdot,h\,\cdot)}{W^{n_0,\infty}(\R^{2d})} \leq \nrm{g}_{W^{n_0,\infty}(\R^{2d})},
		\end{equation*} 
		by taking $g=\nabla_\xi f$, we deduce from Inequality~\eqref{eq:Calderon_Vaillancourt} and the definition of the $\L^\infty$ norm that
		\begin{equation*}
			\Nrm{\Dhv{\Weylh{f}}}{\L^\infty} \leq C \nrm{\nabla_\xi f}_{W^{n_0,\infty}(\R^{2d})},
		\end{equation*} %TODO Can we get something that becomes L^\infty_{x,\xi} when $h\to 0$
		uniformly in $\hbar$. Moreover, taking $g = \Dv f$ in Proposition~\ref{prop:moment_bound} yields
		\begin{equation*}
			\Tr{\n{\Dhv{\Weylh{f}}}\n{\opp}^{n_1}} \leq C \Nrm{\Dv f}{H^\sigma_\sigma(\R^{2d})}.
		\end{equation*}
		The result then follows by combining these two inequalities to bound the right-hand-side of the interpolation inequality~\eqref{eq:interpolation}.
	\end{proof}
	
\section{Proof of Theorem~\ref{thm:CV_Hartree} and Theorem~\ref{thm:Lp_conv}}\label{sec:proof_thm_1_2}

	In this section, we will start by proving a stability estimate similar to the the inequality used in the classical case and then use the results of Section~\ref{sec:regularity} and the propagation of regularity for the Vlasov equation to get the proof of Theorem~\ref{thm:CV_Hartree} and then the proof of Theorem~\ref{thm:Lp_conv}. The conditional result is stated in the following proposition.
	\begin{prop}\label{prop:conditional_estimate}
		Let $K = \frac{1}{\n{x}^a}$ with $a\in\(\(\tfrac{d}{2}-2\)_+,d-1\)$ and assume $\op$ is a solution of the Hartree equation~\eqref{eq:Hartree} with initial condition $\op^\init\in\L^1_+$ and $f\geq 0$ is a solution of the Vlasov equation verifying
		\begin{subequations}
		\begin{align}\label{eq:condition_regu_f}
			f&\in L^\infty_\loc(\R_+, W^{n_0+1,\infty}(\R^{2d})\cap H^{\sigma+1}_{\sigma}(\R^{2d}))
			\\\label{eq:condition_regu_rho_f}
			\rho_f&\in L^\infty_\loc(\R_+,L^1\cap H^\nu)\,,
		\end{align}
		\end{subequations}
		where $n_0=\floor{d/2}+1$, $(n,n_1)\in(2\N)^2$ are such that $n>d/2$ and $n_1\geq \frac{d}{\fb-1}$ and we used the notations $\sigma=2n+n_1$ and $\nu=(n+a+2-d)_+$ and $\fb = \frac{d}{a+1}$. Then the following inequality holds
		\begin{align*}
			\Tr{\n{\op-\op_f}} \leq \(\Tr{\n{\op^\init - \op_f^\init}} + C_f(t)\, \hbar\) e^{\lambda_f(t)},
		\end{align*}
		where
		\begin{align*}
			\lambda_f(t) &= C_{d,n_1,a} \int_0^t \Nrm{\Dv f}{W^{n_0,\infty}(\R^{2d})\cap H^\sigma_\sigma(\R^{2d})} \d s
			\\
			C_f(t) &= C_{d,n_1,a} \int_0^t \Nrm{\rho_f(s)}{L^1\cap H^\nu} \Nrm{\nabla_\xi^2 f(s)}{H^{2n}_{2n}(\R^{2d})} e^{-\lambda_f(s)}\d s.
		\end{align*}
	\end{prop}
	
	\begin{remark}
		It is actually sufficient to assume that $f\geq 0$ holds at $t=0$ since this implies that it holds at any time $t\geq 0$. 
	\end{remark}

	In analogy with the classical case (cf. proof of Proposition~\ref{prop:classical-Lp}), we introduce the two-parameter semigroup $\cU_{t,s}$ such that $\cU_{s,s}=1$ and defined for $t>s$ by 
	\begin{equation*}
		i\,\hbar\,\partial_t\,\cU_{t,s} = H(t)\,\cU_{t,s}\,,
	\end{equation*}
	where $H$ is the Hartree Hamiltonian~\eqref{eq:Hartree-Ham}. We consider the quantity
	\begin{align*}
		 i\,\hbar\,\partial_t \(\cU^*_{t,s}\,(\op(t)-\op_f(t))\,\cU_{t,s}\) =\ &\cU^*_{t,s} \com{K*(\rho(t)-\rho_f(t)),\op_f(t)} \cU_{t,s}
		 \\
		 &+\cU^*_{t,s}\,B_t\ \cU_{t,s}\,,
	 \end{align*}
	 where $B_t$ is an operator defined through its integral kernel by
	 \begin{equation}\label{eq:def_B}
	 	B_t(x,y)=\((K*\rho_f)(x)-(K*\rho_f)(y)-(\nabla K*\rho_f)\left(\dfrac{x+y}{2}\right)\cdot(x-y)\) \r_f(x,y)\,.
	 \end{equation}
	 Using Duhamel's formula and taking the Schatten $p$ norm (recall that $\cU_{t,s}$ is a unitary operator), we get
	 \begin{multline}\label{eq:duhamel-Lp}
	 	\Nrm{\op(t) - \op_f(t)}{p}\leq \Nrm{\op^\init - \op_f^\init}{p} + \dfrac{1}{\hbar}\int_0^t  \Nrm{B_t}{p}\d s
	 	\\
	 	+\dfrac{1}{\hbar}\int_0^t \int |\rho(s,z)-\rho_f(s,z)|\,\Nrm{[K(\cdot-z),\op_f(s)]}{p}\d z\dd s.
	 \end{multline}  
	 We now take $p=1$, i.e. the trace norm, and we have to bound each term on the right-hand side of Inequality~\eqref{eq:duhamel-Lp} in order to obtain a Gr\"{o}nwall type inequality which will prove Proposition~\ref{prop:conditional_estimate}. Note that we will then use again Inequality~\eqref{eq:duhamel-Lp} with $p>1$ together with Theorem~\ref{thm:CV_Hartree} to prove Theorem~\ref{thm:Lp_conv}.
 
\subsection{The commutator inequality}

	Generalizing \cite[Lemma 3.1]{porta_mean_2017}, we obtain the quantum equivalent of Inequality~\eqref{eq:Holder_lorentz}, which is is the following inequality for the trace norm of the commutator of $K$ and a trace class operator $\op$.
	\begin{thm}\label{thm:key_ineq}
		Let $a\in(-1,d-1)$, $K(x) = \frac{1}{\n{x}^a}$ or $K(x)=\ln(\n{x})$ when $a=0$. Then 
		\begin{equation*}
			\nabla K\in L^{\fb,\infty} \ \text{ with }\ \fb = \fb_a := \tfrac{d}{a+1}.
		\end{equation*}
		Let $\fb'$ be the conjugated H\"{o}lder exponent of~$\fb$. Then for any $\eps\in(0,\fb'-1]$, there exists a constant $C>0$ such that
		\begin{equation*}
			\Tr{\n{\com{K(\cdot-z),\op}}} \leq C\, h \Nrm{\Diag{\n{\Dhv{\op}}}}{L^{\fb'-\eps}}^{\frac{1}{2} + \tilde{\eps}} \Nrm{\Diag{\n{\Dhv{\op}}}}{L^{\fb'+\eps}}^{\frac{1}{2} - \tilde{\eps}},
		\end{equation*}
		for any $\tilde{\eps}\in(0,\frac{\eps}{2\fb'})$ and with the additional assumption $\eps < \fb_3'-\fb'$ if $d\geq 4$.
	\end{thm}
	
	\begin{remark}
		In the case of Coulomb interaction and $d=3$, we have $K(x) = \frac{1}{\n{x}}$, $\fb = \fb_1 = \frac{3}{2}$ and $\nabla K\in L^{\frac{3}{2},\infty}$. Thus for any $\eps\in(0,2]$, there exists a constant $C>0$ such that
		\begin{equation*}
			\Tr{\n{\com{K(\cdot-z),\op}}} \leq C\,h \Nrm{\Diag{\n{\Dhv{\op}}}}{L^{3-\eps}}^{\frac{1}{2} + \tilde{\eps}} \Nrm{\Diag{\n{\Dhv{\op}}}}{L^{3+\eps}}^{\frac{1}{2} - \tilde{\eps}},
		\end{equation*}
		for any $\tilde{\eps}\in(0,\frac{\eps}{6})$ 
	\end{remark}
	
	This theorem is a corollary of the slightly more precise following proposition.
	
	\begin{prop}\label{prop:key_ineq}
		For any $\delta\in\(\(\frac{1}{\fb_1'}-\frac{1}{\fb'}\)_+, 1-\frac{1}{\fb'}\)$ and $q\in\lt(\frac{\fb'}{1-\delta\fb'},\infty\rt]$, there exists a constant $C>0$ such that
		\begin{equation}\label{eq:key_ineq}
			\Tr{\n{\com{K(\cdot-z),\op}}} \leq C\, h \Nrm{\Diag{\n{\Dhv{\op}}}}{L^p}^\theta \Nrm{\Diag{\n{\Dhv{\op}}}}{L^q}^{1-\theta},
		\end{equation}
		where $\frac{1}{p} = \frac{1}{\fb'} + \delta$ and $\theta = \delta/\(\frac{1}{p}-\frac{1}{q}\)$ and with the additional assumption that $q<\fb_3'$ if $d\geq 4$.
	\end{prop}
	
	\begin{proof}[Proof of Theorem~\ref{thm:key_ineq}]
		To prove our result, we will decompose the potential as a combination of Gaussian functions (see e.g. \cite[5.9~(3)]{lieb_analysis_2001}). By using the definition of the Gamma function and a simple change of variable, when $a>0$, one obtains the following formula for any $r>0$
		\begin{equation}\label{eq:laplace_transform}
			\frac{1}{\omega_a\,r^{a/2}} = \frac{1}{2}\int_0^\infty t^{\frac{a}{2}-1} e^{-\pi r t} \dd t,
		\end{equation}
		where $\omega_a = \frac{2\,\pi^{a/2}}{\Gamma(a/2)}$.
		Taking $r = \n{x}^2$ directly leads to the following decomposition
		\begin{equation*}
			\frac{1}{\omega_a\n{x}^a} = \frac{1}{2}\int_0^\infty t^{\frac{a}{2}-1} e^{-\pi \n{x}^2 t} \dd t.
		\end{equation*}
		 Now when $a\in(-2,0)$, take Equation~\eqref{eq:laplace_transform} with $a+2$ instead of $a$, integrate it with respect to $r$, exchange the integrals and then replace again $r$ by $\n{x}^2$. This yields a similar decomposition in the form
		\begin{equation*}
			\frac{1}{\omega_a\n{x}^a} = \frac{1}{2}\int_0^\infty t^{\frac{a}{2}-1} \(e^{-\pi \n{x}^2 t}-1\) \dd t.
		\end{equation*}
		In order to treat the case of the logarithm, do the same steps with $a=0$ to obtain
		\begin{equation*}
			-\ln(\n{x}) = \frac{1}{2}\int_0^\infty t^{\frac{a}{2}-1} \(e^{-\pi \n{x}^2 t}-e^{-\pi t}\) \dd t.
		\end{equation*}
		In all these cases, defining $\omega_0 := 1$, we get the following identity
		\begin{align*}
			\frac{1}{\omega_a} \(K(x)-K(y)\) &= \frac{1}{2}\int_0^\infty t^{\frac{a}{2}-1} \(e^{-\pi \n{x}^2 t} - e^{-\pi \n{y}^2 t}\) \dd t.
		\end{align*}
		Following the idea of \cite{porta_mean_2017} but with this new decomposition, we write
		\begin{align*}
			\frac{1}{\omega_a} \(K(x)-K(y)\) &= \frac{1}{2}\int_0^\infty t^{\frac{a}{2}-1} \int_0^1 \frac{\d}{\d\theta} \(e^{-\pi \theta \n{x}^2 t} e^{-\pi \(1-\theta\) \n{y}^2 t}\) \d\theta \dd t
			\\
			&= -\pi\, \int_0^\infty t^{\frac{a}{2}} \int_0^1 \(x-y\)\cdot\(x+y\) e^{-\pi \theta \n{x}^2 t} e^{-\pi \(1-\theta\) \n{y}^2 t} \d\theta \dd t,
		\end{align*}
		from which we get
		\begin{align*}
			\frac{K(x-z)-K(y-z)}{-\pi\,\omega_a} &= \int_0^1\int_0^\infty t^{\frac{a}{2}}  \(x-y\)\cdot\(\phi_\theta(x)\varphi_{1-\theta}(y) + \varphi_{\theta}(x)\phi_{1-\theta}(y)\) \d t \dd \theta,
		\end{align*}
		where we defined $\varphi_k(x) := e^{-k\pi\n{x-z}^2 t}$ and $\phi_k(x) := (x-z)\,\varphi_k(x)$. Thus, since the integral kernel of $\Dhv{\op}$ is $\frac{x-y}{i\hbar}\r(x,y)$ and exchanging $\theta$ by $1-\theta$ in the second term of the integral, we obtain
		\begin{align*}
			\frac{1}{i\pi\hbar\,\omega_a} \lt[K(\cdot-z),\op\rt] &= \int_0^1\int_0^\infty t^\frac{a}{2} \(\phi_{\theta}\cdot \Dhv{\op}\,\varphi_{1-\theta} + \varphi_{1-\theta}\,\Dhv{\op}\cdot \phi_{\theta}\) \d t \dd \theta.
		\end{align*}
		Noticing that $\(\phi_{\theta}\cdot \Dhv{\op}\,\varphi_{1-\theta}\)^* = \varphi_{1-\theta}\,\Dhv{\op}\cdot \phi_{\theta}$, we can now estimate the trace norm by
		\begin{equation}\label{eq:trace_estim_0}
			\frac{1}{h\n{\omega_a}} \Nrm{\lt[K(\cdot-z),\op\rt]}{1} \leq \int_0^1\int_0^\infty t^\frac{a}{2} \Nrm{\phi_{\theta}\cdot \Dhv{\op}\,\varphi_{1-\theta}}{1}\d t \dd \theta.
		\end{equation}
		Then, by decomposing the self-adjoint operator $\Dhv{\op}$ on an orthogonal basis $(\psi_j)_{j\in J}$, we can write $\Dhv{\op} = \sumj \lambda_j \ket{\psi_j}\bra{\psi_j}$ and we get
		\begin{align*}
			\Nrm{\phi_{\theta}\cdot \Dhv{\op}\,\varphi_{1-\theta}}{1} &\leq \sumj \n{\lambda_j} \Nrm{\ket{\phi_\theta\psi_j}\bra{\psi_j \varphi_{1-\theta}}}{1}
			\\
			&\leq \sumj \n{\lambda_j} \Nrm{\phi_\theta\psi_j}{L^2}\Nrm{\psi_j \varphi_{1-\theta}}{L^2},
		\end{align*}
		where we used the fact that $\Nrm{\ket{u}\bra{v}}{1} = \Nrm{u}{L^2} \Nrm{v}{L^2}$. Thus, by the Cauchy-Schwarz inequality for series
		\begin{align*}
			\Nrm{\phi_{\theta}\cdot \Dhv{\op}\,\varphi_{1-\theta}}{1} &\leq \(\sumj \n{\lambda_j} \Nrm{\phi_\theta\psi_j}{L^2}^2\)^\frac{1}{2}\(\sumj \n{\lambda_j} \Nrm{\psi_j \varphi_{1-\theta}}{L^2}^2\)^\frac{1}{2}
			\\
			&\leq \(\intd \n{\phi_{\theta}}^2\rho_1\)^\frac{1}{2} \(\intd \n{\varphi_{1-\theta}}^2\rho_1\)^\frac{1}{2},
		\end{align*}
		with the notation $\rho_1 = \Diag{\n{\Dhv{\op}}} = \sumj \n{\lambda_j} \n{\psi_j}^2$. By the integral H\"{o}lder's inequality, this yields
		\begin{equation}\label{eq:trace_estim_1}
			\Nrm{\phi_{\theta}\cdot \Dhv{\op}\,\varphi_{1-\theta}}{1} \leq \Nrm{\phi_{\theta}}{L^{2p'}} \Nrm{\varphi_{1-\theta}}{L^{2q'}} \Nrm{\rho_1}{L^p}^\frac{1}{2} \Nrm{\rho_1}{L^q}^\frac{1}{2},
		\end{equation}
		where $(p,q)\in[1,\infty]^2$ can depend on the parameter $t$, which will help us to obtain the convergence of the integral in Inequality~\eqref{eq:trace_estim_0}. We can now compute explicitly the integrals of the functions $\phi$ and $\varphi$. It holds
		\begin{align*}
			\Nrm{\phi_{\theta}}{L^{2p'}}^{2p'} &= \intd \n{x-z}^{2p'}e^{-2\pi\theta\n{x-z}^2 p' t} \d x = \frac{\omega_d}{\omega_{d+2p'}} \frac{1}{\(2\theta\, p' t\)^\frac{d+2p'}{2}}
			\\
			\Nrm{\varphi_{1-\theta}}{L^{2q'}}^{2q'} &= \intd e^{-2\pi\(1-\theta\)\n{x-z}^2 q' t} \d x = \frac{1}{\(2\(1-\theta\)q't\)^\frac{d}{2}}.
		\end{align*}
		Combining these two formulas with inequalities~\eqref{eq:trace_estim_0} and~\eqref{eq:trace_estim_1} leads to
		\begin{align*}
			\Nrm{\lt[K(\cdot-z),\op\rt]}{1} \leq h \int_0^\infty \frac{C_{d,a,p'} \Nrm{\rho_1}{L^p}^\frac{1}{2} \Nrm{\rho_1}{L^q}^\frac{1}{2}}{t^{\frac{1}{2}\(\frac{d}{2p'} + \frac{d}{2q'} + 1 - a\)}} \int_0^1 \frac{\d\theta}{\theta^{\frac{d+2p'}{4p'}} \(1-\theta\)^{\frac{d}{4q'}}} \dd t.
		\end{align*}
		with $C_{d,a,p'} = \n{\omega_a} \(\frac{\omega_d}{\omega_{d+2p'}}\)^{\!\frac{1}{2p'}} \(2p'\)^{-\,\frac{d+2p'}{4p'}} \(2q'\)^{-\,\frac{d}{4q'}}$. We observe that the integral over $\theta$ is converging as soon as
		\begin{equation}\label{eq:condition_p_0}
			\frac{1}{p'} < \frac{2}{d} = \frac{1}{\fb_1} \ \text{ and }\ \frac{1}{q'} < \frac{4}{d} = \frac{1}{\fb_3}.
		\end{equation}
		Now, in order to get a finite integral of the variable $t$, we can cut the integral into two parts. The first one for $t\in(0,R)$ and the second one for $t\in(R,\infty)$ for a given $R>0$. Then we have to choose $p$ and $q$ such that 
		\begin{align*}
			\frac{1}{2}\(\frac{d}{2p'} + \frac{d}{2q'} + 1 - a\) < 1 &\text{ for } t\in(0,R)
			\\
			\frac{1}{2}\(\frac{d}{2p'} + \frac{d}{2q'} + 1 - a\) > 1 &\text{ for } t\geq R,
		\end{align*}
		or equivalently since $\fb = \frac{d}{a+1}$
		\begin{align*}
			\frac{1}{2}\(\frac{1}{p'} + \frac{1}{q'}\) < \frac{1}{\fb} &\text{ for } t\in(0,R)
			\\
			\frac{1}{2}\(\frac{1}{p'} + \frac{1}{q'}\) > \frac{1}{\fb} &\text{ for } t\geq R.
		\end{align*}
		However, this has to be compatible with the constraint \eqref{eq:condition_p_0}. Therefore, when $t\in(0,R)$, we can in particular take $q=p_0$ with $p_0 < \min\!\(\fb',\fb_1'\)$. When $t\geq R$, then we can also take for example $p=p_0 > \frac{\fb'}{2}$ and then any $q$ such that
		\begin{equation}\label{eq:condition_q_0}
			 \frac{2}{\fb} - \frac{1}{p_0'} < \frac{1}{q'} < \frac{4}{d} \text{ and } \frac{1}{q'} \leq 1.
		\end{equation}
		Notice that the condition $\frac{1}{q'} < \frac{4}{d}$ is only used when $d\geq 4$ and can be rewritten $q < \fb_3'$. Such a pair $(p_0,q)$ exists as long as $a\leq \frac{d}{2}$ and $a<2$. By defining $\delta := \frac{1}{p_0} - \frac{1}{\fb'}$, then these conditions are equivalent to
		\begin{align*}
			\(\frac{1}{\fb_1'}-\frac{1}{\fb'}\)_+ &< \delta < 1-\frac{1}{\fb'}
			\\
			\frac{1}{p_0} &= \frac{1}{\fb'} + \delta
			\\
			\frac{1}{q} &< \frac{1}{\fb'} - \delta.
		\end{align*}
		With these $p$ and $q$, we therefore deduce that there exists a constant $C$ depending on $d$, $a$, $p_0$ and $q$ such that
		\begin{align*}
			\Nrm{\lt[K(\cdot-z),\op\rt]}{1} \leq C\,h \(R^{\frac{d}{2}\(\frac{1}{\fb} - \frac{1}{p_0'}\)}\Nrm{\rho_1}{L^{p_0}} + R^{\frac{d}{2}\(\frac{1}{\fb} - \frac{1}{2p_0'}-\frac{1}{2q'}\)}\Nrm{\rho_1}{L^{p_0}}^\frac{1}{2}\Nrm{\rho_1}{L^{q}}^\frac{1}{2}\).
		\end{align*}
		Optimizing with respect to $R$ yields
		\begin{equation}\label{eq:trace_estim_2}
			\Tr{\n{\lt[K(\cdot-z),\op\rt]}} \leq C\, h \Nrm{\rho_1}{L^{p_0}}^{\theta_0} \Nrm{\rho_1}{L^q}^{1-\theta_0} 
		\end{equation}
		where $\theta_0 = \frac{1/p_0-1/\fb'}{1/p_0 - 1/q}$. In order to obtain an equation of the form \eqref{eq:key_ineq}, we can define $\eps := q-\fb'$, which is positive by Equation~\eqref{eq:condition_q_0} and the fact that $p_0 < \fb'$. The condition $q < \fb_3'$ when $d\geq 4$ then reads $\eps < \fb_3'-\fb'$. We can also define $p := \fb'-\eps\geq 1$. Then by a direct computation and using again~\eqref{eq:condition_q_0}, we obtain
		\begin{equation*}
			p_0 - p = p_0 + q - 2 \fb' > 0,
		\end{equation*}
		so that $p < p_0 < \fb' < q$ and by interpolation of Lebesgue spaces,
		\begin{equation*}
			\Nrm{\rho_1}{L^{p_0}} \leq \Nrm{\rho_1}{L^{p}}^{\theta_1} \Nrm{\rho_1}{L^{q}}^{1-\theta_1},
		\end{equation*}
		where $\theta_1 = \frac{1/p_0-1/q}{1/p-1/q}$. Noticing that $\theta_0\theta_1 = \frac{1/p_0-1/\fb'}{1/p - 1/q}$ and that we can take $\frac{1}{p_0}$ as close as we want to $\frac{1}{p}$, there exists $\eps_1$ such that we can choose $p_0$ such that $\theta_0\theta_1 + \eps_1 = \frac{1/p-1/\fb'}{1/p - 1/q} = \frac{1}{2} + \frac{\eps}{2\fb'}$. Therefore, the last inequality combined with Inequality~\eqref{eq:trace_estim_2} leads to Formula~\eqref{eq:key_ineq}.
	\end{proof}
	
	The following proposition is an extension of Theorem~\ref{thm:key_ineq} to $\L^p$ spaces, for $p<\fb$. Notice however that the right-side here is expressed in terms of weighted quantum Lebesgue norms, which makes the inequality weaker than the one in Theorem~\ref{thm:key_ineq}.
		
	\begin{prop}\label{prop:commutator_ineq_Lp}
		Let $d\geq2$, $a\in \(-1,\min\!\(2,\frac{d}{2}\)\)$, $1\leq p< \fb := \frac{d}{a+1}$. Then for any $\eps \in (0,q-1)$ and $n> a+1$, there exists a constant $C>0$ such that
		\begin{equation*}
			\Nrm{\lt[K(\cdot-z),\op\rt]}{\L^p} \leq C\, h \Nrm{\Dhv{\op}\,\opm_{n}}{\L^{q+\eps}}^{\frac{1}{2} + \tilde{\eps}} \Nrm{\Dhv{\op}\,\opm_{n}}{\L^{q-\eps}}^{\frac{1}{2} - \tilde{\eps}},
		\end{equation*}
		where $\tilde{\eps} = \eps/q$, $\opm_{n} = 1+\n{\opp}^n$ and with
		\begin{equation*}
			\frac{1}{p} = \frac{1}{q} + \frac{1}{\fb}.
		\end{equation*}
	\end{prop}
	
	\begin{proof} First we do the same decomposition as for the $\L^1$ case but then take a $\L^p$ norm in \eqref{eq:trace_estim_0}. This yields
	\begin{equation}\label{eq:trace_estim_p}
		\frac{1}{h\n{\omega_a}} \Nrm{\lt[K(\cdot-z),\op\rt]}{p} \leq \int_0^\infty \int_0^1 t^\frac{a}{2} \Nrm{\phi_{\theta}\cdot \Dhv{\op}\,\varphi_{1-\theta}}{p}\d\theta \dd t.
	\end{equation}
	In order to bound this integral, we will cut it into two parts corresponding to $t\in(0,R)$ and $t\geq R$, and we take $\frac{1}{q} > \frac{1}{p} - \frac{1}{\fb}$ when $t$ is small, and $\frac{1}{q} < \frac{1}{p} - \frac{1}{\fb}$ in the second case.	Using the hypotheses, we can find $(\alpha,\beta)\in [2,\infty)^2$ and $(n_\alpha,n_\beta) \in (\frac{d}{\alpha},\infty)\times (\frac{d}{\beta},\infty)$ such that $\alpha > d$, $\beta>d/2$, $n_\alpha + n_\beta = n$ and $\frac{1}{\alpha} + \frac{1}{\beta} = \frac{1}{p} - \frac{1}{q}$. Then we define $\opm_k := 1 + \n{\opp}^k$ and multiply and divide by $\opm_{n_\alpha}$ and $\opm_{n_\beta}$. This yields
	\begin{align*}
		\Nrm{\phi_{\theta}\cdot \Dhv{\op}\,\varphi_{1-\theta}}{p} &= \Nrm{\(\phi_{\theta}\,\opm_{n_\alpha}^{-1}\)\cdot \opm_{n_\alpha}\Dhv{\op}\,\opm_{n_\beta}\,\opm_{n_\beta}^{-1}\varphi_{1-\theta}}{p}
		\\
		&\leq \Nrm{\phi_{\theta}\,\opm_{n_\alpha}^{-1}}{\alpha} \Nrm{\opm_{n_\beta}^{-1}\varphi_{1-\theta}}{\beta} \Nrm{\opm_{n_\alpha}\Dhv{\op}\,\opm_{n_\beta}}{q},
	\end{align*}
	where we used twice H\"older's inequality for operators to obtain the third line from the second. We notice that $\phi_{\theta}\,\opm_{n_\alpha}^{-1}$ is of the form $g(-i\nabla)\,f(x)$, so that since $\alpha \geq 2$, by the Kato-Seiler-Simon inequality (see e.g. \cite[Thm 4.1]{simon_trace_2005}), it holds
	\begin{align*}
		\Nrm{\phi_{\theta}\,\opm_{n_\alpha}^{-1}}{\alpha} \leq (2\pi)^{-\,\frac{d}{\alpha}} \Nrm{\phi_{\theta}}{L^\alpha}\Nrm{m_{n_\alpha}^{-1}(\hbar\,\cdot)}{L^\alpha},
	\end{align*}
	with $m_{n_\alpha}\textcolor{blue}{^{-1}}(\hbar x) = \(1+\n{\hbar x}^{n_\alpha}\)^{-1}$. By the change of variable $y = \hbar\,x$ in the last integral, and using the fact that $C_{d,n_\alpha,\alpha} := \Nrm{m_{n_\alpha}^{-1}}{L^\alpha} < \infty$, this yields
	\begin{align*}
		\Nrm{\phi_{\theta}\,\opm_{n_\alpha}^{-1}}{\alpha} \leq C_{d,n_\alpha,\alpha}\, h^{-\,\frac{d}{\alpha}} \Nrm{\phi_{\theta}}{L^\alpha}.
	\end{align*}
	Then a direct computation of the integral of $\phi_{\theta}$ yields
	\begin{align*}
		\Nrm{\phi_{\theta}\,\opm_{n_\alpha}^{-1}}{\alpha} \leq C_{d,n_\alpha,\alpha}\, h^{-\,\frac{d}{\alpha}} \(\frac{\omega_{d}}{\omega_{d+\alpha}}\)^\frac{1}{\alpha} \frac{1}{\(\alpha \theta t\)^{\frac{d+\alpha}{2\,\alpha}}}.
	\end{align*}
	By the same proof but replacing $\phi_\theta$ by $\varphi_{1-\theta}$, if $\beta\geq 2$, we have
	\begin{align*}
		\Nrm{\opm_{n_\beta}^{-1}\,\varphi_{1-\theta}}{\beta} \leq C_{d,{n_\beta},\beta}\, h^{-\,\frac{d}{\beta}} \frac{1}{\(\beta \(1-\theta\) t\)^{\frac{d}{2\,\beta}}}.
	\end{align*}
	Therefore, \eqref{eq:trace_estim_p} leads to
	\begin{align*}
		\Nrm{\lt[K(\cdot-z),\op\rt]}{p} &\leq \int_0^\infty \frac{C_{\op}\,h^{1-d\(\frac{1}{\alpha}+\frac{1}{\beta}\)}}{t^{\frac{1}{2}\(\frac{d}{\alpha}+\frac{d}{\beta}+1-a\)}}  \(\int_0^1 \frac{\d\theta}{\theta^{\frac{d+\alpha}{2\,\alpha}} \(1-\theta\)^{\frac{d}{2\,\beta}}}\) \d t
		\\
		&\leq \int_0^\infty \frac{C_{\op}\,h^{1+\frac{d}{p'}-\frac{d}{q'}}}{t^{\frac{d}{2}\(\frac{1}{p}-\frac{1}{q}-\frac{1}{\fb}\) +1}}  \(\int_0^1 \frac{\d\theta}{\theta^{\frac{d+\alpha}{2\,\alpha}} \(1-\theta\)^{\frac{d}{2\,\beta}}}\) \d t,
	\end{align*}
	where $C_{\op} = \(\frac{\omega_{d}}{\omega_{d+\alpha}}\)^\frac{1}{\alpha} \frac{C_{d,n_\alpha,\alpha} C_{d,n_\beta,\beta}}{\beta^{\frac{d}{2\,\beta}} \alpha^{\frac{d+\alpha}{2\,\alpha}}} \Nrm{\opm_{n_\alpha}\Dhv{\op}\,\opm_{n_\beta}}{q}$. The integrals in $\theta$ and $t$ converge since
	\begin{align*}
		&\alpha > d \ \text{ and }\ \beta > \frac{d}{2}
		\\
		&\frac{1}{p}-\frac{1}{q} < \frac{1}{\fb} \ \text{ if }\ t\in[0,R]
		\\
		&\frac{1}{p}-\frac{1}{q} > \frac{1}{\fb} \ \text{ if }\ t\in(R,\infty).
	\end{align*}
	Then observe that as proved in Appendix~\ref{sec:appendix_B} (Inequality~\eqref{eq:mixing}), it holds
	\begin{equation*}
		\Nrm{\opm_{n_\alpha}\Dhv{\op}\,\opm_{n_\beta}}{q} \leq \Nrm{\Dhv{\op}\,\opm_{n_\beta}\opm_{n_\alpha}}{q} = \Nrm{\Dhv{\op}\,\opm_{n}}{q},
	\end{equation*}
	and we can conclude by taking the optimal $R$ as in the proof of Theorem~\ref{thm:key_ineq}.
	\end{proof}

\subsection{Bound for the error term}

	In this section, we will prove that the operator $B_t$ defined by Equation~\eqref{eq:def_B} is small when $h$ goes to $0$.
	\begin{prop}\label{prop:B-term}
		Under the hypotheses of Theorem~\ref{thm:CV_Hartree}, if $p\in[1,2]$ and $n\in2\N$ with $n>d/2$, then
		\begin{equation*}
			\Nrm{B_t}{\L^p} \leq C\,\hbar^2 \Nrm{\rho_f}{L^1\cap H^\nu} \Nrm{\Dv^2 f}{H^{2n}_{2n}(\R^{2d})},
		\end{equation*}
		where $\nu = (n+a+2-d)_+$ and $C$ is independent from $\hbar$.
	\end{prop}

	\begin{proof}
		Recalling the notation $E = -\nabla K * \rho$, as in \cite{saffirio_semiclassical_2019,saffirio_hartree_2020}, we decompose $B_t$ as follows
		\begin{align*}
			\frac{1}{i\hbar}&B_t(x,y) = \int_0^1 E((1-\theta)x+\theta y) - E\!\(\frac{x+y}{2}\) \d\theta \cdot \Dhv{\op_f}(x,y)
			\\
			&= i\hbar \int_0^1\!\!\int_0^1 \(\!\theta-\frac{1}{2}\)\nabla E\(\!\((1-\theta)x+\theta y\)\theta' +  \frac{x+y}{2}\(1-\theta'\)\!\) \d\theta \dd\theta' : \Dhvv{\op_f}(x,y)
			\\
			&= i\hbar \int_0^1\!\!\int_0^1 \(\!\theta-\frac{1}{2}\) \nabla E\(a_{\theta,\theta'}x + b_{\theta,\theta'} y\) \d\theta \dd\theta' : \Dhvv{\op_f}(x,y)
		\end{align*}
		where $a_{\theta,\theta'} = \frac{\theta'+1}{2}-\theta\theta'$, $b_{\theta,\theta'} = \frac{1-\theta'}{2}+\theta\theta'$ and ":" denotes the double contraction of tensors. In terms of the Fourier transform of $\nabla E$, this yields
		\begin{align*}
			\frac{1}{i\hbar}B_t(x,y) &= i\hbar \int_0^1\!\!\int_0^1 \intd \(\theta-\frac{1}{2}\) e^{2i\pi z\cdot(a_{\theta,\theta'}x + b_{\theta,\theta'} y)}\,\widehat{\nabla E}\(z\) \d\theta \dd\theta'\dd z : \Dhvv{\op_f}(x,y).
		\end{align*}
		Defining $e_\omega$ as the operator of multiplication by the function $x\mapsto e^{2i\pi\omega z\cdot x}$, we obtain
		\begin{align*}
			\frac{1}{i\hbar}B_t &= i\hbar \int_0^1\!\!\int_0^1 \intd \(\theta-\frac{1}{2}\)\widehat{\nabla E}(z) : e_{a_{\theta,\theta'}} \(\Dhvv{\op_f}\) e_{b_{\theta,\theta'}} \d\theta \dd\theta' \dd z,
		\end{align*}
		and since $e_\omega$ is a bounded (unitary) operator, taking the quantum Lebesgue norms yields
		\begin{align*}
			\frac{1}{\hbar}\,\Nrm{B_t}{\L^p} &\leq \hbar \int_0^1\!\!\int_0^1 \intd \n{\theta-\frac{1}{2}} \n{\widehat{\nabla E}(z)} \Nrm{e_{a_{\theta,\theta'}} \(\Dhvv{\op_f}\) e_{b_{\theta,\theta'}}}{\L^p} \d\theta \dd\theta' \dd z
			\\
			&\leq \frac{\hbar}{2} \Nrm{\widehat{\nabla E}}{L^1} \Nrm{\Dhvv{\op_f}}{\L^p}.
		\end{align*}
		
		$\bullet$ Now to bound $\Nrm{\widehat{\nabla E}}{L^1}$, we can use the fact that for any $n > d/2$, the Fourier transform maps $H^n$ continuously into $L^1$ to get
		\begin{align*}
			\Nrm{\widehat{\nabla E}}{L^1} \leq C_{d,n}\Nrm{\nabla^2 K * \rho_f}{H^n}.
		\end{align*}
		If $a=d-2$, then by continuity of $\nabla^2K *\cdot$ in $H^n$, we get $\Nrm{\widehat{\nabla E}}{L^1}\leq C \Nrm{\rho_f}{H^n}$. Else if $a\in\(\frac{d}{2}-2,d\)\setminus\{2\}$
		\begin{align*}
			\Nrm{\widehat{\nabla E}}{L^1} &\leq C_{d,n,a} \Nrm{\(1+\n{x}^{n}\)\n{x}^{a+2-d} \widehat{\rho_f}}{L^2} \leq C_{d,n,a} \Nrm{\rho_f}{\dot{H}^{a+2-d}\cap \dot{H}^{n+a+2-d}}
			\\
			&\leq C_{d,n,a} \Nrm{\rho_f}{L^1\cap H^{(n+a+2-d)_+}}
		\end{align*}
		where we used the fact that if $\alpha\in(-\frac{d}{2},0)$, then by Sobolev's inequalities $ L^{p^*}\subsetarrow \dot{H}^\alpha$ with $\frac{1}{p^*} = \frac{1}{2} - \frac{\alpha}{d}$, and then $L^2\cap L^1 \subsetarrow L^{p^*}$ since ${p^*}\in[1,2]$.
		
		$\bullet$ Finally, to bound $\Nrm{\Dhvv{\op_f}}{\L^p}$, we interpolate it between the $\L^1$ and the $\L^2$ norms to get
		\begin{equation}\label{eq:Weyl_Lp}
			\Nrm{\Dhvv{\op_f}}{\L^p} \leq \Nrm{\Dhvv{\op_f}}{\L^2}^\theta \Nrm{\Dhvv{\op_f}}{\L^1}^{1-\theta} = \Nrm{\nabla_\xi^2 f}{L^2(\R^{2d})}^\theta \Nrm{\Dhvv{\op_f}}{\L^1}^{1-\theta}
		\end{equation}
		where $\theta = 2/p'$. Then using the fact that $\Dhvv{\op_f} = \Weylh{\nabla^2_\xi f}$, we can use  Proposition~\ref{prop:moment_bound} with $g = \nabla^2_\xi f$, $n_1=0$ and $n>d/2$ to get
		\begin{align*}
			\Nrm{\Dhvv{\op_f}}{\L^1} \leq C \Nrm{\nabla^2_\xi f}{H^{2n}_{2n}(\R^{2d})},
		\end{align*}
		which using Inequality~\eqref{eq:Weyl_Lp} implies that $\Nrm{\Dhvv{\op_f}}{\L^p} \leq C \Nrm{\nabla^2_\xi f}{H^{2n}_{2n}(\R^{2d})}$.
	\end{proof}

\subsection{Proof of Proposition~\ref{prop:conditional_estimate}}

	We can now use the bounds on the commutator and the error terms proved in previous sections to prove the stability estimate of Proposition~\ref{prop:conditional_estimate}.

	\begin{proof}[Proof of Proposition~\ref{prop:conditional_estimate}]
		For $p=1$, Equation~\eqref{eq:duhamel-Lp} yields
		\begin{multline*}
			\Tr{\n{\op(t) - \op_f(t)}} \leq \Tr{\n{\op^\init - \op_f^\init}} + \dfrac{1}{\hbar}\int_0^t \Tr{\n{B_s}}\d s
			\\
		 	+\dfrac{1}{\hbar}\int_0^t \int |\rho(s,z)-\rho_f(s,z)|\,\Tr{\n{[K(\cdot-z),\op_f(s)]}}\d z\dd s.
		\end{multline*}
		Proposition~\ref{prop:B-term} gives a bound on the second term on the right-hand side, whereas Theorem~\ref{thm:key_ineq} allows us to bound the last term on the right-hand side uniformly in~$z$. Moreover, because of \eqref{eq:L1-trace} we have 
		\begin{equation*}
			\Nrm{\rho-\rho_f}{L^1}\leq \Tr{\n{\op-\op_f}}\,.
		\end{equation*}
		Altogether, we obtain for some small $\eps>0$ to be chosen later
		\begin{align*}
			\Tr{\n{\op - \op_f}} &\leq \Tr{\n{\op^\init - \op_f^\init}}
			\\
			&\qquad+ C\,\hbar \int_0^t \Nrm{\rho_f(s)}{L^1\cap H^{(n+a+2-d)_+}} \Nrm{\Dv^2 f(s)}{H^{2n}_{2n}(\R^{2d})}\d s
			\\
			&\qquad+ C\,\int_0^t \Tr{\n{\op(s) - \op_f(s)}} \Nrm{\Diag{\n{\Dhv{\op_f(s)}}}}{L^{\fb'+\eps}\cap L^{\fb'-\eps}}\d s
		\end{align*}
		{where $\fb' = \frac{d}{d-\(a+1\)}$}. We then use Proposition~\ref{prop:diag_Dhv_Weyl} to bound the $L^p$ norm of the diagonal for $p=\fb'+\eps$ and $p=\fb'-\eps$
		\begin{align*}
			\Nrm{\Diag{|\Dhv{\op_f}|}}{L^p} \leq \cC_{d,n_1}\Nrm{\Dv f}{W^{n_0,\infty}(\R^{2d})\cap H^{2n+n_1}_{2n+n_1}(\R^{2d})},
		\end{align*}
		where since $n_1 > \frac{d}{\fb-1} = d\(\fb'-1\)$ we can choose $\eps$ such that $\fb'+\eps \leq 1+\frac{n_1}{d}$. We conclude by Gr\"{o}nwall's Lemma.
	\end{proof}

\subsection{Proof of Theorem~\ref{thm:CV_Hartree}}

	\begin{proof}[Proof of Theorem~\ref{thm:CV_Hartree}]
		To prove this theorem, it just remains to prove that the assumptions~\eqref{eq:condition_regu_f} and~\eqref{eq:condition_regu_rho_f} are satisfied with our choice of initial conditions, which will imply the result by Proposition~\ref{prop:conditional_estimate}. But these bounds are only about the solution of the classical Vlasov equation for which the long time existence of regular solutions is known. More precisely, we prove the regularity we need in our case in Proposition~\ref{prop:regu_Vlasov} in Appendix~\ref{sec:appendix_A}. With our assumptions on the initial data, we have $f^\init\in W^{\sigma+1,\infty}_m(\R^{2d})$ with $m>d$. Moreover, since $f^\init\in H^{\sigma+1}_\sigma(\R^{2d})$ with $\sigma>m+\frac{d}{\fb-1}$, by H\"{o}lder's inequality we deduce in particular that $f^\init\in L^2_\sigma(\R^{2d})$ which by H\"{o}lder's inequality yields 
		\begin{equation*}
			\iintd f^\init\n{\xi}^{n_1} \d x\dd\xi < \infty
		\end{equation*}
		for some $n_1>\frac{d}{\fb-1}$. Therefore, Proposition~\ref{prop:regu_Vlasov} indeed leads to
		\begin{align*}
			f &\in L^\infty_\loc(\R_+,W^{\sigma+1,\infty}_m(\R^{2d})\cap H^{\sigma+1}_\sigma(\R^{2d})),
		\end{align*}
		where we notice that $\sigma> n_0 := \floor{\frac{d}{2}} + 1$. Finally, the $H^\nu$ bound for $\rho$ also follows from H\"{o}lder's inequality since $\sigma>d/2$ so that
		\begin{equation*}
			\Nrm{\nabla^{\lceil\nu\rceil}\!\rho}{L^2} \leq \Nrm{\intd \n{\Dx^{\lceil\nu\rceil} f}\d \xi}{L^2} \leq C_{d,\sigma} \Nrm{\weight{\xi}^\sigma \Dx^{\lceil\nu\rceil} f}{L^2(\R^{2d})} \leq C \Nrm{f}{H^{\sigma+1}_\sigma(\R^{2d})},
		\end{equation*}
		where the last inequality follows from the fact that $\lceil\nu\rceil\leq \sigma+1$.
	\end{proof}

\subsection{Proof of Theorem~\ref{thm:Lp_conv}}\label{subsec:proof_thm_2}

	We now prove Theorem~\ref{thm:Lp_conv} using also the results of Propositions~\ref{prop:commutator_ineq_Lp} and \ref{prop:B-term}.
	
	\begin{proof}[Proof of Theorem~\ref{thm:Lp_conv}]
		Recall Equation~\eqref{eq:duhamel-Lp}. The bound~\eqref{eq:L1-trace} yields
		\begin{equation*}
			\begin{split}
			\Nrm{\op - \op_f}{\L^p}&\leq \Nrm{\op^\init - \op_f^\init}{\L^p} + \dfrac{1}{\hbar}\int_0^t \Nrm{B_t}{\L^p}\d s
			\\
			&\qquad+\dfrac{1}{\hbar}\int_0^t \Tr{\n{\op-\op_f}}\,\sup_z\,\Nrm{[K(\cdot-z),\op_f]}{\L^p}\d s.
			\end{split}
		\end{equation*}
		The second term on the right-hand side can be estimated thanks to Proposition~\ref{prop:B-term} and can then be bounded as in the case $p=1$. The last term on the right-hand side is bounded by Proposition~\ref{prop:commutator_ineq_Lp} by terms of the form $\Nrm{\Dhv{\op_f}\,\opm_n}{\L^{q}}$ with $\opm_n = 1+\n{\opp}^n$, $n>a+1=\frac{d}{\fb}$ and $\frac{1}{q}$ close to $\frac{1}{p}-\frac{1}{\fb}$. When $a<\frac{d-2}{2}$, then $q\leq 2$ and we can bound them by interpolating them between $\L^1$ and $\L^2$ weighted norms, yielding
		\begin{equation*}
			\Nrm{\Dhv{\op_f}\,\opm_n}{\L^{q}} \leq \Nrm{\Dhv{\op_f}\,\opm_n}{\L^{2}}^{2/q'} \Nrm{\Dhv{\op_f}\,\opm_n}{\L^{1}}^{1-2/q'},
		\end{equation*}
		and we can bound these terms by Proposition~\ref{prop:multiply_Weyl_L2} and Proposition~\ref{prop:moment_bound}. When $q > 2$, this strategy is no more possible, however by the property of the Weyl transform and Calder\'{o}n-Vaillancourt-Boulkhemair inequality~\eqref{eq:Calderon_Vaillancourt}, we know that
		\begin{align*}
			\Nrm{\Weylh{g}}{\L^2} &= \Nrm{g}{L^2(\R^{2d})}
			\\
			\Nrm{\Weylh{g}}{\L^\infty} &\leq C_d \Nrm{g}{W^{n_0,\infty}(\R^{2d})},
		\end{align*}
		where $n_0 = \floor{\frac{d}{2}}+1$. Therefore, this time, we interpolate between the $\L^2$ and $\L^\infty$ norm to get
		\begin{align}\label{eq:Weyl_Lq}
			\Nrm{\Weylh{g}}{\L^q} &\leq C_d^{\theta} \Nrm{\Weylh{g}}{\L^\infty}^\theta \Nrm{\Weylh{g}}{\L^2}^{1-\theta}
			\leq C_d^{\theta} \Nrm{g}{W^{n_0,\infty}(\R^{2d})}^\theta \Nrm{g}{L^{2}(\R^{2d})}^{1-\theta},
		\end{align}
		where $\theta = 1-\frac{2}{q}$ is close to $\frac{2}{p'}-\frac{1}{\fb'} \pm \eps$. Using Lemma~\ref{lem:multiply_Weyl}, we see that $\Dhv{\op_f}\opm_n$ can be written as a linear combination of terms of the form $\Weylh{\xi^\alpha \partial_x^\beta \Dv f} =: \Weylh{g_{\alpha,\beta}}$ where $\alpha$ and $\beta$ are multi-indices verifying $\n{\alpha+\beta}\leq n$. Therefore, taking $g=g_{\alpha,\beta}$ in inequality~\eqref{eq:Weyl_Lq} for each $g_{\alpha,\beta}$, we obtain a control in terms of weighted Sobolev norms of the solution $f$ of the classical solution of the Vlasov equation~\eqref{eq:Vlasov} of the form
		$\Nrm{f}{W^{\sigma+1,\infty}_\sigma(\R^{2d})\cap H^{\sigma+1}_\sigma(\R^{2d})}$ with $\sigma > n_0 + d/\fb$, which can be controlled as in the proof of Theorem~\ref{thm:CV_Hartree}. We can therefore conclude by Gr\"{o}nwall's Lemma that Inequality~\eqref{eq:Lp_conv} holds.
		
		Now we prove Formula~\eqref{eq:Lp_conv_2}. Consider Equation~\eqref{eq:Lp_conv} and the following bound:
		\begin{equation}
			\Nrm{\op-\op_f}{\L^\infty} \leq \Nrm{\op}{\L^\infty} + \Nrm{\op_f}{\L^\infty}.
		\end{equation}
		As for the first term in the right-hand side, we know that all $\L^p$ norms are propagated by the Hartree equation, therefore $\Nrm{\op}{\L^\infty}=\Nrm{\op^\init}{\L^\infty}$, that is bounded by assumption. In the second term in the right-hand side we use again the Calder\'{o}n-Vaillancourt-Boulkhemair inequality~\eqref{eq:Calderon_Vaillancourt}. Hence, if $f\in W^{n_0,\infty}(\R^{2d})$ and $\op^\init\in\L^\infty$, the $\L^\infty$ norm of the difference $\op-\op_f$ is bounded uniformly in $\hbar$. To conclude, we bound the $\L^q$ norm using the $\L^\infty$ norm and the $\L^p$ norm with $p=\fb-\eps$, for $\eps>0$ small enough, and get
		% 1/q = th/p + (1-th)/infty <=> th = p/q
		\begin{equation*}
			\Nrm{\op-\op_f}{\L^q}\leq \Nrm{\op-\op_f}{\L^p}^\frac{p}{q}\Nrm{\op-\op_f}{\L^\infty}^{1-\frac{p}{q}},
		\end{equation*}
		since $q\in(p,\infty)$. Then Formula~\eqref{eq:Lp_conv} yields 
		\begin{equation*}
			\Nrm{\op-\op_f}{\L^q}\leq C(t)\(\Nrm{\op^\init-\op_f^\init}{\L^p}^\frac{p}{q} + \Tr{\n{\op^\init-\op_f^\init}}^\frac{p}{q} + \hbar^\frac{p}{q}\) e^{\lambda(t)},
		\end{equation*} 
		where $C$ is a constant which depends on the dimension of the space $d$, on $\Nrm{\op^\init}{\L^\infty}$ and on $f^\init$.
	\end{proof}

\section{Proof of Theorem~\ref{thm:CV_Hartree-Fock}}\label{sec:proof_thm_3}

	We recall that $\sfX = \sfX_{\op}$ is the operator of time dependent integral kernel $\sfX_{\op}(x,y) = K\!\(x-y\)\r(x,y)$, where $\r$ is the integral kernel of the operator $\op$. Under the conditions of Theorem~\ref{thm:CV_Hartree-Fock}, the associated energy is bounded and we have the following inequalities
	\begin{prop}\label{prop:exchange-term}
		Let $a\in [0,d)$, $s := d-a$ and $\op$ be a positive trace class operator. Then if $K\in \dot{H}^{s,1}_w$, it holds
		\begin{align}\label{eq:bound_X_energy}
			\Tr{\sfX\op} &\leq C \,h^{s} \Nrm{K}{\dot{H}^{s,1}_w} \Nrm{\n{\opp}^\frac{a}{2}\op\,}{\L^2}^2.
		\end{align}
		Moreover, if $a\in\lt[0,\frac{d}{2}\rt)$ and $K^2\in \dot{H}^{2s-d,1}_w$, then for any $p\in[1,2]$ and $q = \frac{2p}{2-p}\in[2,\infty]$ there exists a constant such that for any operator $\op_2$
		\begin{subequations}
		\begin{align}\label{eq:bound_X_p_small}
			\Nrm{\sfX_{\op}\,\op_2}{\L^p} &\leq C \,h^{s} \Nrm{K^2}{\dot{H}^{2s - d,1}_w}^\frac{1}{2} \Nrm{\n{\opp}^\frac{a}{2}\op\,}{\L^2} \Nrm{\op_2}{\L^q}.
		\end{align}
		When $p\in[2,\infty]$ then we still have
		\begin{align}\label{eq:bound_X_p_big}
			\Nrm{\sfX_{\op}\,\op_2}{\L^p} &\leq C \,h^{s + d\(\frac{1}{p} - \frac{1}{2}\)} \Nrm{K^2}{\dot{H}^{2s - d,1}_w}^\frac{1}{2} \Nrm{\n{\opp}^\frac{a}{2}\op\,}{\L^2} \Nrm{\op_2}{\L^\infty},
		\end{align}
		\end{subequations}
		where in both \eqref{eq:bound_X_p_small} and \eqref{eq:bound_X_p_big} the constants $C$ depend only on $s$ and $d$.
	\end{prop}
	
	\begin{remark}
		We can control the weighted $\L^2$ norms by the following inequality
		\begin{equation}\label{eq:controlling_L2m}
			\Nrm{\n{\opp}^\frac{a}{2}\op\,}{\L^2}^2 \leq \Nrm{\op}{\L^\infty} \Tr{\n{\opp}^a\op}.
		\end{equation}
		Notice that we cannot deduce it immediately by H\"{o}lder's inequality for the Schatten norms because it would give us $\Tr{\n{\n{\opp}^a\op}}$ instead of $\Tr{\n{\opp}^a\op}$ in the right-hand side. However, by definition of the absolute value for operators and by cyclicity of the trace, we get
		\begin{equation*}
			\Nrm{\n{\opp}^\frac{a}{2}\op\,}{2}^2 = \Tr{\op\n{\opp}^a\op} = \Tr{\op\,\op^\frac{1}{2}\n{\opp}^a\op^\frac{1}{2}}  = \Tr{\op\n{\n{\opp}^\frac{a}{2}\op^\frac{1}{2}}^2}.
		\end{equation*}
		Now, H\"{o}lder's inequality gives
		\begin{equation*}
			\Tr{\op\n{\n{\opp}^\frac{a}{2}\op^\frac{1}{2}}^2} \leq \Nrm{\op}{\infty} \Tr{\n{\n{\opp}^\frac{a}{2}\op^\frac{1}{2}}^2} = \Nrm{\op}{\infty} \Tr{\n{\opp}^a\op},
		\end{equation*}
		which leads to \eqref{eq:controlling_L2m} by the definition of $\L^2$ and $\L^\infty$.
	\end{remark}
	
	\begin{proof}[Proof of Proposition~\ref{prop:exchange-term}] We first prove Inequality~\eqref{eq:bound_X_energy} and then use it to show formulas~\eqref{eq:bound_X_p_small} and~\eqref{eq:bound_X_p_big}.
	
	$\bullet$ \textit{Proof of Inequality~\eqref{eq:bound_X_energy}.} Use Formula~\eqref{eq:representation_formula} to get that
		\begin{equation*}
			\iintd K(x-y) \n{\r(x,y)}^2 \d x\dd y = c_{d,a}\intd \(\iintd \frac{\n{\r(x,y)}^2}{\n{x-y-w}^a} \dd x\dd y\) Q(\d w),
		\end{equation*}
		for some measure $Q$ such that $\Nrm{Q}{\mathrm{TV}} = \Nrm{K}{\dot{H}^{s,1}_w}$. This leads to
		\begin{equation*}
			\Eps_\sfX \leq c_{d,a}\sup_{w\in\R^d}\(\iintd \frac{\n{\r(x,y)}^2}{\n{x-y-w}^a} \dd x\dd y\) \Nrm{Q}{\mathrm{TV}}.
		\end{equation*}
		Now we concentrate on bounding the double integral. First we observe that by the change of variable $z=x-y-w$, it holds
		\begin{align*}
			\Eps_a := \iintd \frac{\n{\r(x,y)}^2}{\n{x-y-w}^a} \d x\dd y = \iintd \frac{1}{\n{z}^a} \n{\r(z+y+w,y)}^2 \d z\dd y.
		\end{align*}
		Then, by the Hardy-Rellich inequality (see e.g. \cite{yafaev_sharp_1999}), since $a\in[0,d)$, for any $\varphi\in H^\frac{a}{2}$, it holds
		\begin{align*}
			\intd \frac{|\varphi(z)|^2}{|z|^a} \dd z \leq \cC_{d,a} \intd \n{\Delta^\frac{a}{4}\varphi(z)}^2\d z.
		\end{align*}
		Therefore, taking $\varphi(z) = \r(z+y+w,y)$ in the above inequality and integrating with respect to $y$ yields
		\begin{equation*}
			\Eps_a \leq \cC_{d,a}  \iintd \n{\Delta_z^\frac{a}{4}\r(z+y+w,y)}^2 \d z\dd y = \cC_{d,a}  \iintd \n{\Delta_x^\frac{a}{4}\r(x,y)}^2 \d x\dd y.
		\end{equation*}
		Recalling that $\Delta_x^\frac{a}{4}\r(x,y)$ is nothing but the integral kernel of the operator $\hbar^{-\,\frac{a}{2}}\n{\opp}^\frac{a}{2}\op$ and using the definition of the $\L^2$ norm, we obtain
		\begin{equation*}
			\Eps_a \leq C_{d,a}\,h^{d-a} \Nrm{\n{\opp}^\frac{a}{2}\op\,}{\L^2}^2,
		\end{equation*}
		where $C_{d,a} = (2\pi)^a\,\cC_{d,a}$.
		
		$\bullet$ \textit{Proof of Inequality~\eqref{eq:bound_X_p_small}.} Since $\sfX_{\op}$ is a positive operator, it holds $\sfX_{\op}^2 = \n{\sfX_{\op}}^2$. Moreover, denoting $\tilde{\sfX}_{\op}$ the exchange operator associated to the kernel $K^2$, the following interesting property holds
		\begin{align*}
			\Tr{\sfX_{\op}^2} = \iintd K(x-y)^2 \n{\r(x,y)}^2 \d x\dd y = \Tr{\tilde{\sfX}_{\op}\,\op}. 
		\end{align*}
		From this and H\"{o}lder's inequality for operators, we deduce that if $K^2\in \dot{H}^{2s - d,1}_w$ with $s\in\lt(\frac{d}{2},d\rt]$, then the following inequality holds
		\begin{align*}
			\Nrm{\sfX_{\op}\,\op_2}{p} \leq \Nrm{\op_2}{q} \Nrm{\sfX_{\op}}{2} \leq h^{\frac{d}{q'}} \Nrm{\op_2}{\L^q} \Tr{\tilde{\sfX}_{\op}\,\op}^\frac{1}{2},
		\end{align*}
		which by Formula~\eqref{eq:bound_X_energy} for $K^2$ leads exactly to~\eqref{eq:bound_X_p_small}.
		
		$\bullet$ \textit{Proof of Inequality~\eqref{eq:bound_X_p_big}.} We use the fact that $\Nrm{\sfX_{\op}\,\op_2}{p} \leq \Nrm{\sfX_{\op}\,\op_2}{2}$ for any $p\geq 2$ and then we use~\eqref{eq:bound_X_p_small} for $p=2$ to get
		\begin{align*}
			\Nrm{\sfX_{\op}\,\op_2}{\L^p} \leq h^{\(\frac{d}{2}-\frac{d}{p'}\)} \Nrm{\sfX_{\op}\,\op_2}{\L^2}
			\leq C \,h^{s + d\(\frac{1}{p} - \frac{1}{2}\)} \Nrm{K^2}{\dot{H}^{2s - d,1}_w}^\frac{1}{2} \Nrm{\n{\opp}^\frac{a}{2}\op\,}{\L^2} \Nrm{\op}{\L^\infty}.
		\end{align*}
		The use of the non-semiclassical inequality $\Nrm{\sfX_{\op}\,\op_2}{p} \leq \Nrm{\sfX_{\op}\,\op_2}{2}$ explains the deterioration of the rate, which might not be optimal.
	\end{proof}
	
	When $a<0$, we have similar bounds but using moments in $x$ instead of moments in $\opp$.
	\begin{prop}\label{prop:exchange-term_2}
		Let $a<0$ and $K(x) = \n{x}^{\n{a}}$, then for any positive operator $\op$
		\begin{align}\label{eq:bound_X_energy_2}
			\Tr{\sfX\op} &\leq C \,h^{d} \Nrm{\n{x}^\frac{\n{a}}{2}\op\,}{\L^2}^2.
		\end{align}
		Moreover, for any $p\in[1,\infty]$, there exists a constant $C>0$ such that for any operator $\op_2$
		\begin{subequations}
		\begin{align}\label{eq:bound_X_p_small_2}
			\Nrm{\sfX_{\op}\,\op_2}{\L^p} &\leq C \,h^{d} \Nrm{\n{x}^\frac{\n{a}}{2}\op\,}{\L^2} \Nrm{\op_2}{\L^q} &&\text{when } p\in[1,2)
			\\\label{eq:bound_X_p_big_2}
			\Nrm{\sfX_{\op}\,\op_2}{\L^p} &\leq C \,h^{d\(\frac{1}{p} + \frac{1}{2}\)}  \Nrm{\n{x}^\frac{\n{a}}{2}\op\,}{\L^2} \Nrm{\op_2}{\L^\infty} &&\text{when } p\in[2,\infty],
		\end{align}
		\end{subequations}
		where $q = \frac{2p}{2-p}\in[2,\infty)$ when $p< 2$ and the constants $C$ depend only on $a$ and $d$.
	\end{prop}
	
	\begin{proof}
		The proof of \eqref{eq:bound_X_energy_2} follows simply by writing
		\begin{equation*}
			\iintd K(x-y) \n{\r(x,y)}^2 \d x\dd y \leq C \iintd \(\n{x}^{\n{a}}+\n{y}^{\n{a}}\) \n{\r(x,y)}^2 \d x\dd y
		\end{equation*}
		and observing that the right-hand side is exactly the right-hand side of inequality~\eqref{eq:bound_X_energy_2}. The two other inequalities follow by taking $K^2$ instead of $K$ and H\"{o}lder's inequality as in the proof of Proposition~\ref{prop:exchange-term_2}.
	\end{proof}

	\begin{proof}[Proof of Theorem~\ref{thm:CV_Hartree-Fock}]	
		We proceed as in the proof of Theorem~\ref{thm:CV_Hartree} and consider the one-parameter group of unitary transformations $\cU_t$ generated by the Hartree-Fock Hamiltonian, i.e.
		\begin{equation*}
			i\,\hbar\,\partial_t\,\cU_t=H_{\text{HF}}(t)\,\cU_t\,,
		\end{equation*} 
		and compute 
		\begin{align*}
			i\,\hbar\,\partial_t\(\cU_t^*\(\op- \Weylh{f}\)\cU_t\) = \ &\cU_t^*\com{K*(\rho-\rho_f),\Weylh{f}} \cU_t
			\\
			& + \cU_t^*\,B_t\,\cU_t - \cU_t^*\com{\sfX_{\op},\(\op - \Weylh{f}\)} \cU_t\,.
		\end{align*}
		Using Duhamel formula and taking the $\L^p$ norm using that $\cU_t$ is a unitary operator, we obtain
		\begin{multline}\label{eq:tr-HF}
			\Nrm{\op - \Weylh{f}}{\L^p} \leq \Nrm{\op^\init - \op_f^\init}{\L^p} + \dfrac{1}{\hbar}\int_0^t \Nrm{[K*(\rho-\rho_f),\Weylh{f}]}{\L^p} \d s
			\\
			+\dfrac{1}{\hbar}\int_0^t \Nrm{B_s}{\L^p}\d s+\dfrac{1}{\hbar}\int_0^t \Nrm{[\sfX_{\op},(\op - \Weylh{f})]}{\L^p}\d s.
		\end{multline}
		The second and third terms on the right-hand side in \eqref{eq:tr-HF} can be bounded as in Theorem~\ref{thm:CV_Hartree}. As for the fourth term, we use Proposition~\ref{prop:exchange-term}. 
		
		More precisely, when $K(x) = \pm\n{x}^{-a}$ with $a\in[0,d/2)$, using \eqref{eq:bound_X_p_small} or \eqref{eq:bound_X_p_big}
		\begin{align*}
			\frac{1}{\hbar} \Nrm{[\sfX_{\op},(\op - \Weylh{f})]}{\L^p} &\leq C \,h^{\tilde{s}-1} \Nrm{K^2}{\dot{H}^{d-2a,1}_w}^\frac{1}{2} \Nrm{\n{\opp}^\frac{a}{2}\op\,}{\L^2} \Nrm{\op-\Weylh{f}}{\L^p}
			\\
			&\leq C\,h^{\tilde{s}-1} \Nrm{K^2}{\dot{H}^{d-2a,1}_w}^\frac{1}{2} \Nrm{\n{\opp}^\frac{a}{2}\op\,}{\L^2} \(\Nrm{\op}{\L^p}+\Nrm{\Weylh{f}}{\L^p}\),
		\end{align*}
		with $\tilde{s} = d - a - d\(\frac{1}{2} - \frac{1}{p}\)_+$. When $\tilde{s}\geq 2$, this does not change the order of the rate of convergence $O(h)$. When $\tilde{s}\leq 2$ (i.e. for high values of $a$ and $p$), the contribution of the exchange term becomes bigger than the one of the second term in the right-hand side of Inequality~\eqref{eq:tr-HF}, thus leading to a rate of convergence of the order $O(h^{\tilde{s}-1})$.
		
		When $K(x) = \pm\n{x}^{-a}$ with $a\in(-1,0)$, then we use inequalities~\eqref{eq:bound_X_p_small_2} or \eqref{eq:bound_X_p_big} to get bounds in terms of $\Nrm{\n{x}^\frac{\n{a}}{2}\op\,}{\L^2}$ instead of $\Nrm{\n{\opp}^\frac{\n{a}}{2}\op\,}{\L^2}$.
		
		When $K(x) = \pm \ln(|x|)$, we write that $K(x) \leq C_\eps\(|x|^\eps+|x|^{-\eps}\)$ and use both type of inequalities to get bounds with $\Nrm{\(\n{x}^\frac{\eps}{2} + \n{\opp}^\frac{\eps}{2}\)\op\,}{\L^2}$ instead.
		
		For all the choices of $K$, when $p=1$, we can therefore conclude that
		\begin{equation}\label{eq:tr_HF_gronwall}
			\Nrm{\op - \Weylh{f}}{\L^1}\leq \(\Nrm{\op^\init - \op_f^\init}{\L^1} + C_0(t)\,\hbar+C_1(t)\,\hbar^{s -1}\) e^{\lambda(t)}.
		\end{equation}
		When $p\in(1,\fb)$, we proceed as in the proof of Formula~\eqref{eq:Lp_conv} (the Hartree case) and use Inequality~\eqref{eq:tr_HF_gronwall} to get
		\begin{equation}\label{eq:tr_HF_gronwall_p}
			\Nrm{\op - \Weylh{f}}{\L^p}\leq \Nrm{\op^\init - \op_f^\init}{\L^p} +  C(t)\(\Nrm{\op^\init - \op_f^\init}{\L^1} + \hbar+\hbar^{\tilde{s} -1}\) e^{\lambda(t)}.
		\end{equation}
		Moreover, when $p\in[\fb,\infty)$, again as in the Hartree case, we proceed as in the proof of \eqref{eq:Lp_conv_2}. Following the exact same argument, $\Nrm{\op-\Weylh{f}}{\L^\infty}$ is bounded uniformly in $\hbar$ as soon as $\op^\init\in\L^\infty$ and $f\in W^{\left[\frac{d}{2}\right]+1,\infty}$. Whence,
		\begin{equation*}
			\Nrm{\op-\Weylh{f}}{\L^p}\leq C(t) \(\Nrm{\op^\init - \op_f^\init}{1}^{1-\theta} + C_0\,\hbar^\frac{p}{q} + C_1\,\hbar^{\(\tilde{s}-1\)\frac{p}{q}}\) e^{\lambda(t)}.
		\end{equation*}
		In particular, if $\Nrm{\op^\init - \op_f^\init}{1}\leq C\,\hbar$, we get for the Hilbert-Schmidt norm ($p=2$) a convergence rate $\hbar^{\(\frac{3}{4}-\eps\) \min\lt\{1,s-1\rt\}}$.
	\end{proof}

\appendix
\section{Propagation of weighted Sobolev norms for Vlasov equation}\label{sec:appendix_A}

	The existence of global smooth solutions and the propagation of regularity is a classical result for the Vlasov-Poisson equation. It can be deduced starting from the works of Pfaffelmoser~\cite{pfaffelmoser_global_1992} or Lions and Perthame~\cite{lions_propagation_1991}, which imply the boundedness of the force field, so that any solution with compact support in the phase space will remain compactly supported at any time. Other general results concerning the propagation of regularity can be found in the more recent work \cite{han-kwan_propagation_2019} by Han-Kwan or in Appendix~A in the work \cite{saffirio_semiclassical_2019} by the second author. In our case, we need the boundedness of the solutions of the Vlasov equation in weighted Sobolev norms and we will see that we can propagate norms of the form $W^{\sigma,\infty}_n(\R^{2d})$. We prefer to work in the framework of \cite{lions_propagation_1991} which allows to have non compactly supported solutions which are very interesting physically, since they include for example Gaussian distributions of velocities. Moreover compactly supported solutions are perhaps less pertinent in the context of quantum mechanics. Furthermore, the proof here follows a completely Eulerian point of view. The result of this section is the following.
	\begin{prop}\label{prop:regu_Vlasov}
		Let $K = \frac{1}{\n{x}^a}$ with $a\in(-1,d-2]$ and let $(n,\sigma,n_1)\in\N^3$ be such that $n>d$ and $n_1>\frac{d}{\fb-1}$ with $\fb = \frac{d}{a+1}$. Let $f\geq 0$ be solution of the Vlasov equation~\eqref{eq:Vlasov} with initial data $f^\init \in W^{\sigma,\infty}_n(\R^{2d})$ satisfying
		\begin{align*}
			\iintd f^\init \n{\xi}^{n_1} \dd x\dd \xi < \infty.
		\end{align*}
		Then the following regularity estimates hold
		\begin{subequations}
		\begin{align}\label{eq:regu_f}
			f &\in L^\infty_\loc(\R_+,W^{\sigma,\infty}_n(\R^{2d}))
			\\\label{eq:regu_rho_f}
			\nabla^\sigma\rho_f &\in L^\infty_\loc(\R_+,L^\infty).
		\end{align}
		\end{subequations}
		If in addition $f^\init\in H^\sigma_k(\R^{2d})$ for some $k\in\R_+$, then
		\begin{equation*}
			f \in L^\infty_\loc(\R_+,H^\sigma_k(\R^{2d})).	
		\end{equation*}
	\end{prop}
	
	The proof works in two steps. We first explain in the next lemma how to get a control of the regularity as soon as $\rho_f$ is uniformly bounded. Then we finish the proof of the theorem by proving that this assumption on $\rho_f$ holds.

	\begin{lem}
		Let $f$ be a solution of the Vlasov equation~\eqref{eq:Vlasov} as in Proposition~\ref{prop:regu_Vlasov} with $\sigma\geq 1$ and assume moreover that
		\begin{equation}\label{eq:assumption_rho_f}
			\rho_f \in L^\infty_\loc(\R_+,L^\infty\cap L^1).
		\end{equation}
		Then the regularity estimates~\eqref{eq:regu_f} and \eqref{eq:regu_rho_f} hold.
	\end{lem}
	
	\begin{proof}
		For clarity, we first start with the case $\sigma=1$ for which we will do a detailed proof, and we will then explain how to modify the proof to get higher regularity estimates. We follow the strategy explained in the course notes \cite{golse_mean_2013}.
		
		\step{1. Case $\sigma=1$} Define the transport operator $\sfT := \xi\cdot\nabla_x + E\cdot\nabla_\xi$. Then, it holds
		\begin{subequations}
		\begin{align}\label{eq:dt_Dx}
			\dpt(\nabla_x f) &= -\sfT \nabla_x f - \nabla E \cdot \nabla_\xi f
			\\\label{eq:dt_Dv}
			\dpt(\nabla_\xi f) &= -\sfT\nabla_\xi f - \nabla_x f.
		\end{align}
		\end{subequations}
		To simplify the computations, recall that $\sfT^* = -\sfT$ and $\sfT(uv) = u\sfT(v) + \sfT(u)v$. Hence, by writing $m_n := 1+\n{\xi}^{np} + \n{x}^{np}$ and using the notation $u^p := |u|^{p-1}u$, it holds
		\begin{equation}\label{eq:ipp_m_0}
			\iintd \sfT(u)\cdot u^{p-1}\,m_n = -\iintd u\cdot\sfT(u^{p-1})\,m_n + \n{u}^p\sfT(m_n).
		\end{equation}
		However noticing that
		\begin{align*}
			u\cdot(\sfT(u^{p-1})) &= u\cdot\(|u|^{p-2}\sfT(u) + \(p-2\)\(\sfT(u)\cdot u\)u^{p-3}\)
			\\
			&= u^{p-1}\cdot\sfT(u) + \(p-2\)\(\sfT(u)\cdot u\)\n{u}^{p-2}
			\\
			&= \(p-1\)u^{p-1}\cdot\sfT(u).
		\end{align*}
		We can simplify Equation~\eqref{eq:ipp_m_0} into
		\begin{align}\label{eq:ipp_m}
			-p\iintd \sfT(u)\cdot u^{p-1}m_n &= \iintd \n{u}^p\sfT(m_n).
		\end{align}
		Now define
		\begin{align*}
			M_x := \iintd \n{\Dx f}^p m_n
			&\text{ and } M_\xi := \iintd \n{\Dv f}^p m_n.
		\end{align*}
		Then using \eqref{eq:dt_Dx} and Formula~\eqref{eq:ipp_m} for $u=\Dx f$ leads to
		\begin{align*}
			\ddt{M_x} &= -p\iintd (\Dx f)^{p-1}\cdot\(\sfT \Dx f + \nabla E\cdot\Dv f\)m_n
			\\
			&\leq \iintd \n{\Dx f}^p \sfT(m_n) + \Nrm{\nabla E}{L^\infty}(M_\xi + (p-1)M_x)
		\end{align*}
		where we used the multiplicative Young's inequality $p\,a\,b^{p-1} \leq a^p+\(p-1\)b^p$ to get the second term. In the same way, using \eqref{eq:dt_Dv} and taking $u=\Dv f$ yields
		\begin{align*}
			\ddt{M_\xi} &\leq \iintd \n{\nabla_\xi f}^p \sfT(m_n) + (M_x + (p-1)M_\xi).
		\end{align*}
		Then again by the multiplicative Young's inequality
		\begin{align*}
			\sfT(m_n) = np\(E\cdot \xi^{np-1} + \xi\cdot x^{np-1}\) \leq np\(\Nrm{E}{L^\infty} + 1\) m_n.
		\end{align*}
		Thus, for $M_{x,\xi} := M_x + M_\xi$, we obtain
		\begin{align*}
			\dt M_{x,\xi} \leq p\(n\Nrm{E}{L^\infty} + 1+\Nrm{\nabla E}{L^\infty}\) M_{x,\xi}.
		\end{align*}
		However, since we know that $\rho_f\in L^\infty_\loc(\R_+,L^\infty\cap L^1)$ by assumption, we also get the following control on $\Nrm{E}{L^\infty}$
		\begin{equation*}
			\Nrm{E}{L^\infty} \leq C\(\Nrm{\rho_f}{L^\infty} + \Nrm{\rho_f}{L^1}\) \leq C_t
		\end{equation*}
		for some function of time $C_t$ locally bounded on $\R_+$. To control $\nabla E$, we can use the integral Young's inequality if $\nabla K$ is less singular than the Coulomb potential (i.e. if $a<d-2$), and if $a=1$, then we use a singular integral estimate in the spirit of the one in~\cite{beale_remarks_1984} which can be found in the course notes \cite{golse_mean_2013} and can be written
		\begin{align*}
			\Nrm{\nabla E}{L^\infty} &\leq C\(1 + M_0 + \Nrm{\rho_f}{L^\infty} \ln(1+\Nrm{\nabla\rho_f}{L^\infty})\)
			\\
			&\leq C_t\(1+\ln(1+\Nrm{\nabla \rho_f}{L^\infty})\) =: J(t).
		\end{align*}
		Combining these bounds we arrive at $\dt M_{x,\xi} \leq p\(1+n\) J(t) M_{x,\xi}$ which by Gr\"{o}nwall's Lemma implies
		\begin{align*}
			M_{x,\xi}^\frac{1}{p}(t) \leq M_{x,\xi}^\frac{1}{p}(0) \,e^{(1+n) \int_0^tJ}.
		\end{align*}
		Now, since $M_{x,\xi}^\frac{1}{p}$ is equivalent to $\Nrm{f}{W^{1,p}_n(\R^{2d})}$ in the sense that each one is bounded by above by the other up to a multiplicative constant, letting $p\to\infty$, we obtain
		\begin{align}\label{eq:bound_nabla_f}
			\Nrm{f}{W^{1,\infty}_n(\R^{2d})} \leq \Nrm{f^\init}{W^{1,\infty}_n(\R^{2d})}\,e^{(1+n) \int_0^t J}.
		\end{align}
		However, since $n>d$, we have
		\begin{equation}\label{eq:bound_nabla_rho}
			\n{\nabla \rho_f} \leq \intd \n{\nabla_x f}\d\xi \leq C_{d,n}\Nrm{f}{W^{1,\infty}_n(\R^{2d})}
		\end{equation}
		where $C_{d,n} = \intd \weight{\xi}^{-n}\d\xi < \infty$. Combining the two inequalities~\eqref{eq:bound_nabla_f} and~\eqref{eq:bound_nabla_rho} and the fact that $e^{J(t)}\geq 1$, we deduce
		\begin{align*}
			J(t) &\leq C_t + C_t\,\ln\!\(\(1+C_{d,n}\Nrm{f^\init}{W^{1,\infty}_n(\R^{2d})}\)e^{(1+n)\int_0^t J}\)
			\\
			&\leq C_t + C_t\,\ln\!\(1+C_{d,n}\Nrm{f^\init}{W^{1,\infty}_n(\R^{2d})}\) + C_t \(1+n\) \int_0^t J.
		\end{align*}
		Hence, by Gr\"{o}nwall's Lemma
		\begin{align*}
			J(t) &\leq J(0) + \frac{1+\ln\!\(1+C_{d,n}\Nrm{f^\init}{W^{1,\infty}_n(\R^{2d})}\)}{n+1}\frac{e^{C_t(1+n)t}}{1+n}.
		\end{align*}
		We then deduce the bounds on $\Nrm{f}{W^{1,\infty}_n(\R^{2d})}$ and $\nabla\rho_f$ by inequalities~\eqref{eq:bound_nabla_f} and~\eqref{eq:bound_nabla_rho}.
		
		\step{2. Case $\sigma>1$} We give details for $\sigma=2$. The generalization to $\sigma \geq 2$ follows in the same way. In the case $\sigma=2$, Formulas~\eqref{eq:dt_Dv} and \eqref{eq:dt_Dx} become
		\begin{align*}
			\dpt(\Dv^2 f) + \sfT\Dv^2 f &= -2\,\Dx\Dv f
			\\
			\dpt(\Dx^2 f)+\sfT\Dx^2 f &= -2\,\nabla E\cdot\Dv\Dx f -\nabla^2 E\cdot \Dv f\,.
		\end{align*}
		Moreover, the mixed derivative of order two solves
		\begin{equation*}
			\dpt(\Dx\Dv f)+\sfT\Dx\Dv f= -\Dx^2 f - \nabla E\cdot\Dv^2 f\,. 
		\end{equation*}
		We define the quantities
		\begin{align*}
			& M_{xx} := \iintd \n{\Dx^2 f}^p m_n\dd x\dd\xi\,;
			\\
			& M_{\xi\xi} := \iintd \n{\Dv^2 f}^p m_n\dd x\dd\xi\,;
			\\
			& M_{x\xi} := \iintd \n{\Dx\Dv f}^p m_n\dd x\dd\xi\,;
		\end{align*}
		and compute their time derivatives, using the multiplicative Young's inequality, the bound on $\sfT(m_n)$ as in Step~1 and the fact that $p>1$:
		\begin{align*}
			\ddt{M_{\xi\xi}} &\leq p\(n\Nrm{E}{L^\infty}+n+2\) M_{\xi\xi} + 2\,p\,M_{x\xi}\,,
			\\\\
			\ddt{M_{x\xi}} &\leq p\(n\Nrm{E}{L^\infty}+n+1+\Nrm{\nabla E}{L^\infty}\) M_{x\xi} + p\,M_{xx} + p\Nrm{\nabla E}{L^\infty} M_{\xi\xi}\,,
			\\\\
			\ddt{M_{xx}} &\leq p\(n\Nrm{E}{L^\infty} + n + 2 \Nrm{\nabla E}{L^\infty}+\Nrm{\nabla^2 E}{L^\infty}\) M_{xx}
			\\
			&\qquad + 2\,p \Nrm{\nabla E}{L^\infty} M_{x\xi} + p \Nrm{\nabla^2 E}{L^\infty} M_{\xi}\,,
		\end{align*}
		where $M_{\xi}$ is defined and bounded as in Step~1. Thus, for $M_2 := M_{xx} + M_{x\xi}+M_{\xi\xi},$ we obtain
		\begin{align*}
			\dt M_2 \leq Cp\left(n\Nrm{E}{L^\infty}+n+2+2\Nrm{\nabla E}{L^\infty}+\Nrm{\nabla^2 E}{L^\infty}\right)\,M_2
		\end{align*}
		We proved in Step~1 that $\Nrm{E}{L^\infty}$ and $\Nrm{\nabla E}{L^\infty}$ are bounded. To control $\nabla^2 E$, we proceed analogously to Step~1. More generally, we can bound $\nabla^\sigma E$ by $\Dx^\sigma f$. This leads, by Gr\"{o}nwall's Lemma, to
		\begin{align}\label{eq:Lp_Dxv_s}
			M_{2}^\frac{1}{p}(t) \leq M_{2}^\frac{1}{p}(0) \,e^{C_t},
		\end{align}
		for some positive time dependent constant $C_t>0$. Now, since $M_{2}^\frac{1}{p}$ is equivalent to $\Nrm{f}{W^{2,p}_n(\R^{2d})}$ (with the exact same meaning given in Step~1), letting $p\to\infty$, we obtain
		\begin{align*}
			\Nrm{f}{W^{2,\infty}_n(\R^{2d})} \leq \Nrm{f^\init}{W^{2,\infty}_n(\R^{2d})}\,e^{C_t}.
		\end{align*}
		The general case $\sigma>1$ can be handled analogously by defining 
		\begin{equation*}
			M_{\sigma}:=\iint \n{\nabla^\sigma f}^p m_n\dd x\dd\xi\,,
		\end{equation*}
		where $\sigma = |\boldsymbol{\sigma}|$ stands for the sum of the components of the multi-index $\boldsymbol{\sigma}=(\sigma_1,\sigma_2,\dots)$.
	\end{proof}
	
	\begin{proof}[Proof of Proposition~\ref{prop:regu_Vlasov}]
		It just remains to prove that Assumption~\eqref{eq:assumption_rho_f} holds. First notice that the method used in \cite[Theorem~1]{lions_propagation_1991} actually works for any $a\in(-1, d-2]$ since the Coulomb potential is decomposed in two parts of the form $\nabla K \in L^{3/2,\infty}\cap L^1 + W^{2,\infty}$. This proves that the $n_1$ moments can be propagated, which implies that $\rho_f \in L^p$ for $p=1+\frac{n_1}{d}$ by the kinetic interpolation inequality. Then, by Young's inequality, since $n_1> \frac{d}{\fb-1}$, we deduce that
		\begin{equation*}
			E \in L^\infty_\loc(\R_+,L^\infty).
		\end{equation*}
		Finally, as proved in \cite[Corollary~5.1]{lafleche_propagation_2019}, this bound combined with the initial assumption $f\in L^\infty(1+\n{\xi}^n)$ is sufficient to control $\Nrm{\rho_f}{L^\infty}$ and gives
		\begin{align*}
			\Nrm{\rho_f(t)}{L^\infty} \leq C \(1+\int_0^t \Nrm{E(s)}{L^\infty} \d s\),
		\end{align*}
		which implies~\eqref{eq:assumption_rho_f} so that we can apply the lemma. Then once we know the $W^{s,\infty}_n(\R^{2d})$ norm is bounded at any time, if the $H^\sigma_k(\R^{2d})$ is also initially bounded, we can use again Formula~\eqref{eq:Lp_Dxv_s} but with $p=2$ and then bound the terms involving $E$ and $\Dx f$ by the $W^{\sigma,\infty}_n(\R^{2d})$ norm. We conclude again by Gr\"{o}nwall's Lemma.
	\end{proof}
	
\section{Operator identities}\label{sec:appendix_B}

	We list here some formulas for operators, which are used in this paper. First, if $A$ and $B$ are self-adjoint
	\begin{equation}\label{eq:cyclicity_p}
		\Nrm{AB}{p} = \Nrm{BA}{p}
	\end{equation}
	which follows from the fact that the singular values are the same for an operator and its adjoint \cite[Formula~1.3]{simon_trace_2005}. Then we shall remember H\"{o}lder's inequality for operators \cite[Theorem~2.8]{simon_trace_2005} which tells that for any bounded operator $A$ and $B$ and any $(p,q,r)\in[1,\infty]^3$ such that $\frac{1}{p}=\frac{1}{q}+\frac{1}{r}$, it holds
	\begin{equation}\label{eq:Holder}\tag{H\"{o}lder}
		\Nrm{AB}{p} \leq \Nrm{A}{q} \Nrm{B}{r}.
	\end{equation}
	The second important inequality is the Araki-Lieb-Thirring inequality \cite[Theorem~1]{araki_inequality_1990} which states that for any operator $A,B\geq 0$ and any $(q,r)\in[1,\infty)\times\R_+$, the following inequality is true
	\begin{equation*}
		\Tr{(BAB)^{qr}} \leq \Tr{(B^qA^qB^q)^{r}}.
	\end{equation*}
	Replacing $A$ by $A^2$ and observing that $\n{AB}^2 = BA^2B$, this can be rewritten
	\begin{equation}\label{eq:ALT_2}
		\Nrm{AB}{qr}^q \leq \Nrm{A^qB^q}{r}.
	\end{equation}
	These inequalities show that regrouping operators together in Schatten norms increases the value of the norm, while \textit{mixing} them will lower the value. In the same spirit, for any $A, B\geq 0$, $p\geq 1$ and $r\geq 0$, the following mixing inequality holds
	\begin{equation}\label{eq:mixing}
		\Nrm{B^rAB}{p} \leq \Nrm{AB^{r+1}}{p}
	\end{equation}
	
	\begin{proof}[Proof of Inequality~\eqref{eq:mixing}]
		By \ref{eq:Holder}'s inequality, we have
		\begin{align*}
			\Nrm{B^rAB}{p} \leq \Nrm{B^rA^\frac{r}{r+1}}{\frac{r+1}{r} p} \Nrm{A^\frac{1}{r+1} B}{\(r+1\) p}.
		\end{align*}
		Now, by the cyclicity property~\eqref{eq:cyclicity_p} and by Inequality~\eqref{eq:ALT_2}, we get
		\begin{align*}
			\Nrm{B^rA^\frac{r}{r+1}}{\frac{r+1}{r} p} &\leq \Nrm{AB^{r+1}}{p}^\frac{r}{r+1}
			\\
			\Nrm{A^\frac{1}{r+1} B}{\(r+1\) p} &\leq \Nrm{AB^{r+1}}{p}^\frac{1}{r+1},
		\end{align*}
		which yields the result.
	\end{proof}
	\medskip

{\bf Acknowledgements.}
C.S. acknowledges the support of the Swiss National Science Foundation through the Eccellenza project PCEFP2\_181153 and of the NCCR SwissMAP.
%

%% ********************  Bibliographie  ********************

\renewcommand{\bibname}{\centerline{Bibliography}}
\bibliographystyle{abbrv} % apalike, ieee, plain, alpha, unsrt, abbrv
\bibliography{../Vlasov}

\end{document}